\documentclass[11pt]{article}
\usepackage[a4paper,margin=1in]{geometry}
\usepackage[T1]{fontenc}
\usepackage[utf8]{inputenc}
\usepackage[USenglish]{babel}
\usepackage{lmodern}
\usepackage{microtype}
\usepackage{amsmath,amssymb,amsthm,mathtools}
\usepackage{algorithm}
\usepackage[noend]{algpseudocode}
\usepackage{enumitem}
\usepackage{xspace}
\usepackage{xcolor}
\usepackage{graphicx}
\usepackage{hyperref}
\usepackage[nameinlink,noabbrev]{cleveref}
\usepackage[most]{tcolorbox}

\hypersetup{
  colorlinks=true,
  linkcolor=blue!50!black,
  citecolor=blue!50!black,
  urlcolor=blue!50!black
}

\newcommand{\LOCAL}{\textup{\textsc{LOCAL}}\xspace}
\newcommand{\LLL}{\textup{\textsc{LLL}}\xspace}
\newcommand{\poly}{\operatorname{poly}}

\newcommand{\E}{\mathbb{E}}
\newcommand{\Prob}{\mathbb{P}}

\newcommand{\eps}{\varepsilon}

\newtheorem{theorem}{Theorem}
\newtheorem{lemma}[theorem]{Lemma}
\newtheorem{proposition}[theorem]{Proposition}
\newtheorem{corollary}[theorem]{Corollary}
\theoremstyle{definition}
\newtheorem{definition}[theorem]{Definition}

\title{Triangle-Free Coloring in LOCAL via Resilient Lovász Local Lemma}

\author{Peter Davies-Peck, Xusheng Zhang}

\begin{document}

\maketitle

\begin{abstract}
The Lov\'asz Local Lemma (LLL) is a probabilistic tool that has been shown to be of central importance in the study of distributed algorithms. For example, the constructive LLL is known to be complete for the class of locally-checkable labeling problems with $o(\log n)$ randomized complexities in the \textsc{LOCAL} model. One classic application of the LLL is in coloring graphs with some sparse structure, such as triangle-free graphs. Triangle-free coloring therefore serves as a benchmark problem for techniques for sublogarithmic randomized distributed algorithms. 

The state-of-the-art distributed triangle-free coloring algorithm of Pettie and Su [ICALP 2013, Information and Computation 2015] uses $\frac{\Delta}{k}$ colors (where $k$ can be up to $(\frac14 - \varepsilon)\ln \Delta$) and consists of $O(k+\log^* n)$ applications of the distributed LLL. However, the distributed LLL is itself a difficult problem; despite significant study, the fastest algorithms known require $O(\log_\Delta n)$ or $O(\frac{\Delta}{\log\Delta})+\log^{O(1)}\log n$ rounds, and no $\log^{o(1)} n$-round algorithms exist except for restrictive special cases or graph classes such as trees. These LLL applications therefore dominate the round complexity of the algorithm.

In this work, we adapt the Pettie-Su's algorithm so that 
the resulting LLL instances can be solved in $\log^{O(1)}\log n$ rounds, by employing the `resilience' definition of Davies [SODA 2023]. This gives an $O(k)+ \log^{O(1)}\log n$ complexity (since the LLL is not needed when $k= \log^{\omega(1)}\log n$), essentially causing the LLL steps to no longer be the bottleneck of the algorithm. As a corollary we obtain the first $\log^{O(1)}\log n$-round algorithms for coloring triangle-free graphs with $o(\Delta)$ colors. The same framework also yields a companion girth-$5$ algorithm, using $(1+\varepsilon)\Delta/\ln \Delta$ colors in $O(k)+ \log^{O(1)}\log n$ rounds, matching the best known existential upper bound for the number of colors.

The LLL instances we use are among the most difficult for which efficient \textsc{LOCAL} algorithms have been shown, and we hope that this sheds further light on the solvability of the distributed LLL and may lead to further efficient applications.
\end{abstract}

\section*{Acknowledgements and AI Disclosure}
This work is supported by EPSRC New Investigator Award UKRI155 ``Distributed Lov\'asz Local Lemma''.

Codex was used to assist with language editing and \LaTeX{} typing. The tool affected the presentation of the manuscript text and source formatting. 
The authors take full responsibility for the content of this work.

\section{Introduction}

Graph coloring in the \LOCAL model is one of the benchmark symmetry-breaking problems in distributed graph algorithms. 
In the \LOCAL model, each vertex is a processor, and in each synchronous round it can exchange arbitrarily large messages with its neighbors; consequently, an $r$-round algorithm can use only the information contained in the radius-$r$ neighborhood of each vertex; 
we refer to Section~\ref{sec:preliminaries} for the formal setup. Linial's foundational work~\cite{Linial92} established this model as a clean way to study locality without conflating it with bandwidth limits or local computation costs. For general graphs of maximum degree $\Delta$, much of the literature has focused on $(\Delta+1)$-coloring and related degree-plus-one list-coloring variants. This target is optimal in general: complete graphs and odd cycles require $\Delta+1$ colors. There is a long line of literature on $(\Delta+1)$-coloring algorithms, culminating in extremely fast shattering-based algorithms ~\cite{Johansson99,BarenboimElkinKuhn14,HarrisSchneiderSu18,ChangLiPettie20,HKNT22}. Sparse graph classes raise a different question: can one exploit additional structure not merely to color quickly, but also to asymptotically beat the trivial $(\Delta+1)$ palette size?

Triangle-free graphs are among the cleanest examples of this phenomenon. On the existential side, the best current bound on the number of colors possible is $(1+o(1))\Delta/\ln \Delta$~\cite{Molloy19,Bernshteyn19}. On the distributed side, however, the best near-optimal results still stem from the iterative palette-reduction framework of~\cite{PettieSu15}. That framework achieves a $\frac{\Delta}{k}$ coloring, where $k$ can be up to $(\frac14 - \eps)\ln \Delta$, and consists of $O(k+\log^* n)$ randomized phases. If $\Delta>(\ln n)^{\Theta_\eps(1)}$, then those phases succeed with high probability, and so the overall round complexity is $O(k+\log^* n)$. However, for lower $\Delta$, an application of the distributed LLL is required to fix the randomness for those phases. If one plugs in the best currently-known distributed LLL algorithms, therefore, one gets a round complexity of 

\begin{itemize}
\item $O(\log n)$ using the LLL algorithm of Chung, Pettie, and Su~\cite[Theorem 9]{ChungPettieSu17};
\item $O(\frac{k\Delta}{\log\Delta} +k \log^{O(1)}\log n) = O(\Delta) + \log^{O(1)}\log n$ using the LLL algorithm of Davies~\cite{Davies23}.
\end{itemize}

Our contribution is to adapt the framework to give LLL instances that can be solved in $\log^{O(1)}\log n$ rounds, thereby obtaining an $O(k)+\log^{O(1)} \log n$-round algorithm. For the regime $\Delta = \log^{\eps} n$ with constant $\eps \in (0,1)$, which is the most difficult regime for the LLL, this represents a nearly exponential improvement in round complexity.

\subsection{Main results.}

\begin{theorem}
There is a distributed algorithm in the \LOCAL model that, for any $\eps>0$, 
any $\Delta = \Omega_\eps(1)$, and any $k\le (\frac14 - \eps)\ln\Delta$, 
with high probability colors every triangle-free graph of maximum degree $\Delta$ using $\Delta/k$ colors in $O_{\eps}(k)+ \log^{O(1)} \log n$ rounds.
\label{thm:triangle-free-main}
\end{theorem}

By applying Theorem \ref{thm:triangle-free-main} with $k= \min\{\log^{O(1)}\log n, \frac15\ln \Delta\}$, we can color triangle-free graphs with some $o(\Delta)$ number of colors very quickly:

\begin{corollary}
There is a distributed algorithm in the \LOCAL model that, with high probability, colors every triangle-free graph of maximum degree $\Delta$ using $o(\Delta)$ colors in $\log^{O(1)} \log n$ rounds.
\end{corollary}

This complexity is polynomially close to tight, since there is an $\Omega(\log\log n)$ lower bound even for the problem of coloring constant-degree trees with $\Delta$ or fewer colors~\cite{BrandtEtAl16}. The number of colors matches the best known for any (non-trivial, i.e. $o(n)$-round) distributed algorithm. Molloy \cite{Molloy19} showed that triangle-free graphs can be $(1+o(1))\frac{\Delta}{\ln\Delta}$-colored, but it is unknown whether this can be done in a distributed fashion. For girth-$5$ graphs, however, this $(1+o(1))\Delta/\ln \Delta$ result is achievable with the same round complexity:

\begin{theorem}
    \label{thm:girth 5}
There is a distributed algorithm in the \LOCAL model that, for any $\eps>0$, 
any $\Delta = \Omega_\eps(1)$, and any $k\le (1 - \eps)\ln\Delta$, 
with high probability, $(1+o(1))\Delta/k$-colors every girth-$5$ graph of maximum degree $\Delta$ in  $O_{\eps}(k)+ \log^{O(1)} \log n$ rounds.
\end{theorem}

Throughout the paper, the notation \(\Delta=\Omega_\eps(1)\) means that after fixing \(\eps\), 
we assume \(\Delta\ge \Delta_\eps\) for a sufficiently large threshold $\Delta_\eps$ depending only on \(\eps\). 
and we treat \(\eps\) as an arbitrarily small but fixed positive constant.
We expect that our arguments could also be adapted to the setting where \(\eps=\eps(\Delta)\) tends to \(0\) sufficiently slowly, 
by tracking these dependencies more explicitly.

Both the triangle-free and girth-$5$ results extend to the corresponding list-coloring setting. 
Indeed, the algorithms and proofs are phrased entirely in terms of evolving palettes and their sizes. 
Thus the same running times hold when each vertex \(u\) is given an arbitrary initial palette \(P_0(u)\) of size at least \(\Delta/k\).

Our analysis is based on the following general idea, which is not a formal or well-defined statement but provides the intuition for the approach.

\begin{tcolorbox}[enhanced, sharp corners, drop fuzzy shadow=gray]
\textbf{Informal Claim.} 
LLL instances in which the bad events are that some quantities are too far from their expectation, where

\begin{itemize}
\item the quantities are primarily influenced by nodes' immediate neighborhoods, and
\item we can tolerate $O(\frac{\Delta}{r})$ additive error in the quantities,
\end{itemize}
 
should be $O(r)$-resilient and therefore $O(r)+\log^{O(1)}\log n$-round solvable by the algorithm of \cite{Davies23}.
\end{tcolorbox}

Coloring applications of the LLL often have properties along these lines, 
since error in concentration bounds can often be offset by allowing a corresponding number of additional colors. 
For the triangle-free coloring algorithm of \cite{PettieSu15}, for instance, if we allow an additional additive $O(\frac{\Delta}{\poly\log\Delta})$ error in the bounds proven on degree, 
palette size, and color degree throughout the algorithm, then we cannot show that degrees drop below $O(\frac{\Delta}{\poly\log\Delta})$. 
We therefore require $O(\frac{\Delta}{\poly\log\Delta})$ additional colors to color the remaining uncolored graph, but this is negligible in the overall color bound. The potential applicability of this idea is not limited to coloring problems: problems requiring that the frequencies of certain output labels in a node's vicinity are within certain thresholds, while easier than general LLL instances, constitute the majority of concrete applications of the LLL in distributed algorithms.

As we will see, though, incorporating the bad events in the LLL instances for triangle-free coloring into the resilience framework of \cite{Davies23} is highly nontrivial and requires substantial changes to the analysis and some technical changes to the algorithm of \cite{PettieSu15}.

\subsection{Technical Overview}
Our starting point is the iterative coloring process of~\cite{PettieSu15}. 
In each round, that algorithm maintains a palette for every uncolored vertex and 
analyzes how palette sizes and average $c$-degrees evolve under a random selection process, following the average-$c$-degree viewpoint already used in earlier triangle-free coloring analyses such as~\cite{JamallColoring11}, where a $c$-degree counts how many neighbors can still have color $c$ in their palette. 
Our contribution is to speed up the distributed \LLL step that resolves the local failures left by this random step. 
For that purpose, we analyze the resulting \LLL instances using the resilient partition framework of~\cite{Davies23}. Informally, resilience means that a bad event stays unlikely even if one adversarially chooses a subset of variables from a single part of a partition to be resampled, while the new values themselves come from fresh randomness; see Section~\ref{sec:resilient-lll} for the formal definition.

In that framework, the \LLL instance is equipped with a partition \(\Phi=\{\Phi_1,\dots,\Phi_r\}\) of the bad events, and the algorithm processes one part \(\Phi_j\) at a time. The partition is useful only when each event has few neighbors in each part, because then only few neighboring events can sample and possibly revert simultaneously in any one step, which keeps the number of possible reversion patterns under control and makes resilience provable. 

Our technical goal is to prove that the bad events arising from the coloring process are resilient with respect to a suitable partition. 
The main complication is that these bad events involve the \(c\)-degrees of one-hop neighbors, 
so in the resilience proof one must control how resampling a single part propagates through two-hop neighborhoods. 
On top of this, triangle-free graphs still have dependencies that mean that individual \(c\)-degree does not decrease `smoothly', 
so  the analysis here and in \cite{PettieSu15} instead analyzes the \emph{average} \(c\)-degree of a node's neighbors. For girth-$5$ graphs, the same resilient interface leads to a shorter argument, because one can track palette-size, \(c\)-degree, and residual-degree events for each node directly and thereby recover the sharper constant-$1$ regime.

The key new object for our setting is therefore a \emph{color-balanced partition} for the active random choices of the coloring process; see Section~\ref{sec:color-balanced-partition} for the formal definition. It is chosen so that for each vertex \(u\), each color \(c\), and each part \(\Phi_j\), the vertices in \(\Phi_j\) contribute only a \(1/r\)-share of the \(c\)-degree of \(u\), up to constant factors. This is the extra control needed for the color-by-color failure modes of the coloring process. Once the corresponding \LLL instance is shown to be resilient with respect to this partition, the available distributed \LLL algorithm yields an $O(r)+\log^{O(1)} \log n$-round algorithm, where \(r\) is the number of parts.

\subsection{Previous work.}
For sparse coloring, classical existential work already shows that triangle-free and girth-$5$ graphs admit far fewer than $\Delta+1$ colors. Ajtai, Koml\'os, and Szemer\'edi initiated the triangle-free line with the bound $\chi(G)=O(\Delta/\ln \Delta)$~\cite{AjtaiKomlosSzemeredi80}, Johansson later proved the corresponding asymptotic list-coloring bound~\cite{JohanssonTriangleFree96}, Kim established the sharper $(1+o(1))\Delta/\ln \Delta$ asymptotic for girth-$5$ graphs~\cite{Kim95}, and the best current target for the number of colors in triangle-free graphs is now $(1+o(1))\Delta/\ln \Delta$~\cite{Molloy19,Bernshteyn19}. Classical constructions of Kostochka and Bollob\'as show that, even for graphs of arbitrarily large girth, the chromatic number can be as large as $\Omega(\Delta/\log \Delta)$~\cite{Kostochka78,Bollobas78}, so the $\Delta/\ln\Delta$ scale is best possible up to a constant factor. For distributed algorithms, Grable and Panconesi gave efficient randomized Brooks--Vizing colorings for triangle-free and girth-$5$ graphs under additional regularity and degree assumptions~\cite{GrablePanconesi00}, and Barenboim, Elkin, Pettie, and Schneider studied fast distributed coloring more broadly in related sparse settings, including trees and sufficiently high-girth graphs~\cite{BarenboimEPS12}. The framework of~\cite{PettieSu15} gives distributed algorithms achieving both the triangle-free bound $(4+o(1))\Delta/\ln \Delta$ and the girth-$5$ bound $(1+o(1))\Delta/\ln \Delta$. Our work should be viewed as improving the distributed complexity of that framework while retaining the same asymptotic color guarantees in the two settings.

In our setting, the distributed Lov\'asz Local Lemma is the subroutine that handles the local bad events caused by the random coloring step. 
More broadly, however, distributed \LLL has become a central benchmark problem in the \LOCAL model; 
for example, it is the canonical complete problem for a major randomized \LOCAL complexity class on bounded-degree graphs~\cite{ChangPettie19}. 
The constructive resampling framework of~\cite{MoserTardos10} already implies a distributed algorithm via parallel resampling, and Chung, Pettie, and Su gave the first explicit distributed treatment, including $O(\log n)$-round algorithms under standard and strengthened criteria as well as an $\Omega(\log^* n)$ lower bound~\cite{ChungPettieSu17}. 
Brandt et al.\ later proved an $\Omega(\log \log n)$ lower bound even under much stronger criteria~\cite{BrandtEtAl16}. Fischer and Ghaffari introduced the shattering-based sublogarithmic approach~\cite{FischerGhaffari17}, which, when complemented by later distributed derandomization and modern network-decomposition tools~\cite{GhaffariHarrisKuhn18,RozhonGhaffari20,GhaffariGrunauRozhon21}, yielded an $O(d^2 + \log^{O(1)}\log n)$-round algorithm. 
Davies \cite{Davies23} improved this round complexity for the general LLL to $O(\frac{d}{\log d}+ \log^{O(1)}\log n)$, and also developed the resilient partition framework, which underlies our use.

There exists an $O(\log \log n)$-round LLL algorithm on trees \cite{CHLPU19}, an algorithm with $O(\log \log n)$-round \emph{average-case} complexity \cite{Davies-Peck25}, and $\log^{O(1)}\log n$-round algorithms for certain special cases of the LLL that exhibit particularly nice structures (e.g. \cite{HMN22, HM24}), such as that bad events can be avoided based on subsets of their variables regardless of the values taken by the other variables. However, no $\log^{1-\Omega(1)} n$-round worst-case complexity algorithm for the general LLL has yet been found. In this work we aim to widen the types of LLL instances for which $\log^{O(1)}\log n$ algorithms are known; the instances we use are more complex than previous efficiently-solvable special cases, and are the first special cases to be directly shown to fit the $r$-resilience framework of \cite{Davies23} with $r = \omega(1)$. Previously, $r$-resilience with $r = \omega(1)$ was only used by \cite{Davies23} as a stepping-stone to the general LLL result.

\section{Preliminaries}\label{sec:preliminaries}

\subsection{Model}

We work in the randomized \LOCAL model on an $n$-vertex graph $G=(V,E)$ of maximum degree $\Delta$. Each vertex $v\in V$ hosts a processor with a unique identifier. Communication proceeds in synchronous rounds, and in each round every vertex may perform arbitrary local computation, exchange arbitrarily large messages with its neighbors, and sample fresh private random bits. At the end of the computation, each vertex must output its own part of the global solution.

Thus, after $r$ rounds, the state of a vertex is determined entirely by its radius-$r$ neighborhood together with the identifiers, local inputs, and random choices revealed there. The complexity measure is the number of communication rounds. In particular, when we say that an algorithm colors $G$ in $t$ rounds, we mean that after $t$ rounds every vertex outputs a color and the resulting labeling is a proper coloring with the claimed guarantee. All success probabilities are taken over the internal randomness of the algorithm.

\subsection{Local Lemma Instances and the Polynomially-Weakened Regime}

An \LLL instance consists of a family $\mathcal X$ of bad events defined on a family $\mathcal V$ of mutually independent random variables. For each event $A\in\mathcal X$, let $\mathrm{var}(A)\subseteq \mathcal V$ denote the variables on which $A$ depends. The \emph{dependency graph} $G_{\mathcal X}$ has vertex set $\mathcal X$, with two events adjacent whenever their variable sets intersect. Its maximum degree is the dependency parameter $d$ of the instance. In the symmetric setting, $p$ denotes a common upper bound on the bad-event probabilities. When convenient, we write an \LLL instance as a tuple $\mathfrak L=(\mathcal X,\mathcal V)$.

\subsubsection{The polynomially-weakened criterion.}

The classical symmetric \LLL criterion is $epd\le 1$, where $p$ is the maximum bad-event probability and $d$ is the maximum dependency degree. In this paper, we work throughout in the more restrictive regime $p(ed)^c<1$ for a large constant $c$. Despite the terminology, this polynomially-weakened criterion is stricter than the traditional \LLL threshold. We will use two algorithms for this regime, depending on additional structure.

\subsubsection{An algorithm for resilient instances.}\label{sec:resilient-lll}

We use the resilient local lemma framework of~\cite{Davies23}. 
Following the notation and definitions there, we work with two independent copies $\mathcal V^1,\mathcal V^2$ of the variable family $\mathcal V$ and an $r$-partition $\Phi=\{\Phi_1,\dots,\Phi_r\}$ of the events, that is, a partition into $r$ parts. Since, in our applications, both events and variables will be explicitly associated with a particular node of the input graph (for example, events will be things such as ``node $v$'s degree reduces sufficiently'', and variables will be things such as ``the color picked by node $v$''), we will sometimes abuse notation by considering $\Phi$ to be a partition of the variables; this should be understood to mean that the variables in $\Phi_i$ are those associated with the nodes/events in $\Phi_i$. This is to avoid having to use the variable allocation notation in~\cite{Davies23}. 

For an event $A\in\mathcal X$ and a set of variables $S\subseteq \mathcal V$, let $A^S$ denote the event obtained by evaluating $A$ using the second sample on the variables in $S$ and the first sample on all other variables. Define
\[
A'=\bigcup_{i\le r,\;S\subseteq \Phi_i}
\left\{\Prob_{\mathcal V^2}\left[A^S\right]\ge d^{-10}\right\},
\]
where $d$ is the dependency parameter of the instance. The instance is called $\Phi$-resilient if for every event $A$,
\[
\Prob_{\mathcal V^1}[A']\le d^{-30}.
\]
Intuitively, $\Phi$-resilience means that even if one adversarially chooses an arbitrary subset of a single partition part to switch from the first sample to the second, no bad event becomes likely. We say that an instance is \emph{$r$-resilient} if it admits an $r$-partition $\Phi$ for which it is $\Phi$-resilient.
This is an additional structural assumption on top of the polynomially-weakened condition above; in particular, resilience is stricter than the polynomially-weakened condition alone.

\begin{theorem}[Davies~\cite{Davies23}, Theorem~12]
Any $r$-resilient \LLL instance, provided with an $r$-partition $\Phi$ of its variables for which it is $\Phi$-resilient, can be solved in $O(r+\log^{O(1)}\log n)$ rounds in randomized \LOCAL, with high probability.
\label{thm:davies12}
\end{theorem}

\subsection{Coloring the Remaining Graph}
After coloring the majority of the graph using the adapted framework of \cite{PettieSu15} coupled with our analysis of the resulting LLL instances, we will show that the maximum degree $d_{\mathrm{stop}}$ of the remaining uncolored graph is $o(\frac{\Delta}{\log\Delta})$ and then color it using $O(d_{\mathrm{stop}})$ fresh colors.

To do so, we can use the algorithm of Chang, Li, and Pettie \cite{ChangLiPettie20} running in $O(\log^* n + \text{Det}(\poly\log n)$ rounds, where $\text{Det}(n')$ is the deterministic complexity of (deg+1)-list coloring on $n'$-vertex graphs. Using the current fastest such deterministic algorithm, of \cite{GG24}, gives a round complexity of $\tilde O(\log^{5/3}\log n)$.

\begin{theorem}[$\Delta+1$-list-coloring \cite{ChangLiPettie20,GG24}]
Any graph of maximum degree $\Delta$, equipped with lists of size at least $\Delta+1$ at each vertex, can be list-colored in $\tilde O(\log^{5/3}\log n)$ rounds in randomized \LOCAL, succeeding with high probability.
\label{thm:clp20}
\end{theorem}

\section{Triangle-Free Graphs}\label{sec:triangle-free}

This section proves the triangle-free case through a round-based coloring framework. 
In each round, uncolored vertices make local random choices, conflicts are removed, and palettes are updated; 
we formulate each round as an \LLL instance whose purpose is to preserve some round-based invariants. 

This application of the resilient \LLL algorithm, however, is not automatic.
There are several layers of difficulty.
In each round, the iterative algorithm must maintain three invariants, so we must define three corresponding families of bad events and avoid them simultaneously.
Moreover, these events depend on not only by their one-hop neighbors but also by their two-hop neighbors.
In triangle-free graphs, these two-hop effects do not directly decompose into independent contributions to different one-hop neighbors, which makes the analysis more delicate.
Furthermore, to apply the resilient Lovász Local Lemma algorithm, the resilience condition must hold uniformly for every subset \(S\) of variables inside any given part of the partition.
The concentration bounds available to us are not strong enough to support a naive union bound over all such \(S\), so instead for each bad event we identify a ``symptom event'' that can be triggered by any choice of \(S\).
What makes the application even harder is that all of these requirements must be handled by one and the same partition, 
and such a partition also needs to be constructed fast.

The rest of the section addresses these issues.
\Cref{sec:triangle-free-round-instance} defines the round-\(i\) instance, \Cref{sec:color-balanced-partition} gives the partition tool used by the resilient solver, 
\Cref{sec:triangle-free-iterative-algorithm} presents the iterative algorithm, \Cref{sec:triangle-free-parameters} specifies the bad events and parameter schedule, and \Cref{sec:triangle-free-main-lemmas} assembles the lemmas to prove Theorem~\ref{thm:triangle-free-main}. Throughout Sections~\ref{sec:triangle-free} and~\ref{sec:triangle-free-lemma-proofs}, all asymptotic inequalities are interpreted for sufficiently large \(\Delta = \Omega_\eps(1)\); equivalently, after fixing \(\eps\), we assume \(\Delta\ge \Delta_\eps\).

\subsection{The Round-\texorpdfstring{$i$}{i} LLL Instance}
\label{sec:triangle-free-round-instance}

We now make the first step of the roadmap precise by considering one generic round \(i\) and describing exactly what the local \LLL solver receives. The process is monotone: once a vertex is colored, it remains colored and plays no role in later rounds. Therefore, round \(i\) only sees the vertices that are still uncolored after rounds \(1,\dots,i-1\); we call them \emph{active at the start of round \(i\)}.

Fix an iteration \(i\ge 1\). Let \(G_{i-1}\) be the graph induced by vertices that are still uncolored after rounds \(1,\dots,i-1\). 
Let \(P_{i-1}(u)\) be the palette of an active vertex \(u\), and let \(N_{i-1}(u)\) be its residual neighborhood. For any color \(c\), define the color-specific neighborhood
\[
N_{i-1,c}(u)=\{v\in N_{i-1}(u): c\in P_{i-1}(v)\}.
\]
We call the vertices in \(N_{i-1,c}(u)\) the \(c\)-neighbors of \(u\), and call \(|N_{i-1,c}(u)|\) the corresponding \(c\)-degree. The local random step is given in Algorithm~\ref{alg:select}.

\begin{algorithm}[H]
\caption{$\mathrm{Select}(u,\pi_i,\beta_i)$}
\label{alg:select}
\begin{algorithmic}[1]
\ForAll{$c\in P_{i-1}(u)$ independently}
  \State include $c$ in $\mathrm{Sel}_i(u)$ with probability $\pi_i$
  \State $r_c\gets \beta_i/(1-\pi_i)^{|N_{i-1,c}(u)|}$
  \If{$r_c\le 1$}
    \State include $c$ in $K_i(u)$ with probability $r_c$
  \EndIf
\EndFor
\State \Return $(\mathrm{Sel}_i(u),K_i(u))$
\end{algorithmic}
\end{algorithm}

Here \(\mathrm{Sel}_i(u)\) is the set of colors tentatively selected by \(u\), 
while \(K_i(u)\) is the set of colors kept available for \(u\) before conflicts are removed. 

For each active vertex \(u\), define the round-\(i\) random variable \(X_i(u):=(\mathrm{Sel}_i(u),K_i(u))\). The variable family of round \(i\) is therefore \(\mathcal V_i=\{X_i(u): u\in V(G_{i-1})\}\).
Each variable \(X_i(u)\) is sampled independently at its own vertex by Algorithm~\ref{alg:select}.
Let \(\theta_i=(\pi_i,\beta_i,\alpha_i,p_i,p_i',t_i,t_i',d_i,d_i')\) denote the deterministic parameter vector used in round \(i\). 
We treat these quantities as part of the round-\(i\) \LLL input, even though their exact formulas are only specified later in Section~\ref{sec:triangle-free-parameters}. 
With \(\mathcal B_i\) denoting the local bad events also specified later, the round-\(i\) local \LLL instance is \(\mathfrak L_i=(\mathcal B_i,\mathcal V_i,\theta_i)\).

\subsection{\texorpdfstring{\(r\)-Partitions with Degree Parameters}{r-Partitions with Degree Parameters}}\label{sec:color-balanced-partition}

The resilient \LLL theorem from Section~\ref{sec:resilient-lll} requires an \(r\)-partition of the event family. Since our events and variables are both indexed by vertices, such a partition is equivalently a partition of the vertex set, and can also be considered a partition of the variable set.

In \cite{Davies23}, the partitions that are used for the general LLL result are chosen to have an upper bound on node degree within each part of the partition. Here we need slightly more complex partitions that also have an upper bound on c-degrees.

We now specify formally what we require from these partitions. We will use c-degree as defined in Section \ref{sec:triangle-free-round-instance}, but to avoid tying the following Definition \ref{def:r-partition-degree-params} to a specific round of the algorithm we will define notation without the subscript $i$. Let \(H\) be a graph on a vertex set \(U\), and suppose that each \(u\in U\) has a palette \(P(u)\). For a color \(c\), write
\[
N_c(u)=\{v\in N_H(u):c\in P(v)\}.
\]
The useful partitions are those in which no single part captures too much of either kind of degree: for every vertex \(u\), each part should contain only an $O(1/r)$-share of the neighbors of \(u\) and of the \(c\)-neighbors of \(u\) for each color \(c\).

\begin{definition}
\label{def:r-partition-degree-params}
Let \(d,t>0\), and let \(\Phi=\{\Phi_1,\dots,\Phi_r\}\) be a partition of \(U\). 
We call \(\Phi\) a \emph{(color-balanced) $r$-partition with parameters \(d\) and \(t\)} if:
\begin{enumerate}[leftmargin=*,itemsep=0.25em]
  \item for every \(u\in U\) and every part \(\Phi_j\), one has \(|N_H(u)\cap \Phi_j|\le 100d/r\);
  \item for every \(u\in U\), every color \(c\in P(u)\), and every part \(\Phi_j\), one has \(|N_c(u)\cap \Phi_j|\le 100t/r\).
\end{enumerate}
\end{definition}

Here \(d\) is a max-degree upper bound and \(t\) is a \(c\)-degree upper bound. Thus the first condition controls max-degree contributions, and the second controls \(c\)-degree contributions.

\begin{theorem}[\texorpdfstring{\(r\)-partition with degree parameters}{r-partition with degree parameters}]
\label{thm:color-balanced-partition}
For every fixed \(a>0\), there is \(b>0\) such that the following holds.
Assume \(\Delta\le(\log n)^a\), let \(H\) be a graph on at most \(n\) vertices with maximum degree at most \(\Delta\), and suppose that \(|P(u)|\le \Delta\) for every \(u\in U\). Set \(r=(\log\Delta)^{20}\), rounded to an integer if necessary, and let \(0<d\le\Delta\) and \(t>0\). Suppose that
\[
|N_H(u)|\le d
\qquad\text{and}\qquad
|N_c(u)|\le t
\]
for every \(u\in U\) and every color \(c\in P(u)\). Suppose also that
\[
d\ge (\log\Delta)^b
\qquad\text{and}\qquad
t\ge (\log\Delta)^b.
\]
Then, with high probability, the vertex set \(U\), equivalently the variables indexed by \(U\), admits an \(r\)-partition with parameters \(d\) and \(t\). Such a partition can be constructed in \(\log^{O(1)}\log n\) rounds in the \LOCAL model.
\end{theorem}

We construct the partition through an auxiliary LLL instance: each vertex in \(U\) chooses a part, the bad events are the violations of the two conditions from the definition, and these bad events are \(1\)-resilient, so the resilient LLL algorithm from Section~\ref{sec:resilient-lll} applies.

\begin{proof}
For each vertex \(v\in U\), introduce an independent variable \(Y_v\), uniformly distributed on \([r]\). The value \(Y_v=j\) means that \(v\) is put into the part \(\Phi_j\). For an outcome \(Y\), write \(\Phi_j(Y)=\{v\in U:Y_v=j\}\).
We define an auxiliary LLL instance whose bad events are
\[
A_u=\left\{\exists j\in[r]\text{ such that }|N_H(u)\cap \Phi_j(Y)|>100d/r\right\}
\]
for each \(u\in U\), and
\[
B_{u,c}=\left\{\exists j\in[r]\text{ such that }|N_c(u)\cap \Phi_j(Y)|>100t/r\right\}
\]
for each \(u\in U\) and color \(c\in P(u)\). Avoiding all these events is exactly the definition of an \(r\)-partition with parameters \(d\) and \(t\).

We first bound the bad-event probabilities. For a fixed part \(j\), the random variable \(|N_H(u)\cap \Phi_j(Y)|\) is binomial with mean at most \(d/r\). By Chernoff's bound,
\[
\Prob\left[|N_H(u)\cap \Phi_j(Y)|>100d/r\right]\le \exp(-30d/r).
\]
Taking a union bound over \(j\in[r]\), we get \(\Prob[A_u]\le r\exp(-30d/r)\le \Delta^{-200}\) after choosing \(b\) large enough. Similarly, for a fixed part \(j\), the random variable \(|N_c(u)\cap \Phi_j(Y)|\) has mean at most \(t/r\), so Chernoff's bound and a union bound over the parts give \(\Prob[B_{u,c}]\le r\exp(-30t/r)\le \Delta^{-200}\).
Next we count dependencies. Both \(A_u\) and \(B_{u,c}\) use only variables \(Y_v\) with \(v\in N_H(u)\). If one of these events shares a variable with \(A_x\) or \(B_{x,c'}\), then \(x\) is within distance two of \(u\) in \(H\). There are at most \(1+\Delta^2\) such choices of \(x\). For each \(x\), there is one event \(A_x\) and at most \(|P(x)|\le \Delta\) events \(B_{x,c'}\). Thus the dependency degree of the auxiliary instance is at most
\[
(1+\Delta^2)(1+\Delta)\le 4\Delta^3.
\]
Therefore \(p(e\cdot 4\Delta^3)^{32}<1\) for the probability bound \(p=\Delta^{-200}\).

It remains to check \(1\)-resilience of this auxiliary instance. The partition used here is not the output partition \(\Phi(Y)\) that we are trying to construct. Since we only need \(1\)-resilience, we use the trivial one-part partition of the variables \(\{Y_v:v\in U\}\). Thus, in the definition of resilience, \(S\) may be any subset of these variables: the variables in \(S\) are read from the second sample, and all other variables are read from the first sample.

Fix one bad event, and write it in the common form
\[
E=\left\{\exists j\in[r]\text{ such that }|W\cap \Phi_j(Y)|>T\right\},
\]
where either \(W=N_H(u)\) and \(T=100d/r\), or \(W=N_c(u)\) and \(T=100t/r\). In both cases, \(|W|/r\le T/100\).

Let \(Y^1,Y^2\) be the two samples from the definition of resilience, and let \(S\) be any subset of variables. In the mixed event \(E^S\), the variables outside \(S\) keep their first-sample values and the variables in \(S\) use their second-sample values. If the first sample satisfies
\[
|W\cap \Phi_j(Y^1)|\le T/2
\qquad\text{for every }j\in[r],
\]
then, for every choice of \(S\), the second sample must contribute more than \(T/2\) vertices \(v\in W\) whose variable \(Y_v\) belongs to \(S\) to one part in order for \(E^S\) to occur. By the same Chernoff bound and a union bound over the \(r\) parts, this has probability at most \(\Delta^{-200}\), after increasing \(b\) if necessary. Since the dependency degree is at most \(4\Delta^3\), this is at most \((4\Delta^3)^{-10}\), which is the \(D^{-10}\) threshold from the resilience definition (Section~\ref{sec:resilient-lll}) with dependency parameter \(D=4\Delta^3\).

Therefore the exceptional first-sample event in the definition of resilience can occur only if
\[
|W\cap \Phi_j(Y^1)|>T/2
\qquad\text{for some }j\in[r].
\]
The probability of this event is at most \(\Delta^{-200}\), again by the same Chernoff bound and a union bound over the \(r\) parts. Since the dependency degree is at most \(4\Delta^3\), this is at most \((4\Delta^3)^{-30}\), which is the \(D^{-30}\) threshold from the resilience definition with dependency parameter \(D=4\Delta^3\). Hence every bad event in the auxiliary instance is \(1\)-resilient.

The resilient LLL algorithm from Theorem~\ref{thm:davies12}, applied with this one-part variable partition, finds an assignment avoiding all \(A_u\) and \(B_{u,c}\) with high probability in \(\log^{O(1)}\log n\) rounds. Let \(\Phi=\{\Phi_1,\ldots,\Phi_r\}\), where \(\Phi_j=\{v:Y_v=j\}\). Since no bad event occurs, \(\Phi\) is an \(r\)-partition with parameters \(d\) and \(t\).
\end{proof}

\subsection{The Iterative Algorithm}
\label{sec:triangle-free-iterative-algorithm}

Recall from Section~\ref{sec:triangle-free-round-instance} that the active vertices at the start of round \(i\) are exactly the vertices of \(G_{i-1}\). Such a vertex \(u\) has a palette \(P_{i-1}(u)\), a residual neighborhood \(N_{i-1}(u)\), and the color-specific neighborhoods \(N_{i-1,c}(u)\). The round-\(i\) \LLL instance defined there contains the local choices \(\mathrm{Sel}_i(u)\) and \(K_i(u)\) used by the update below.

Algorithm~\ref{alg:triangle-free} presents this coloring framework.
For \(i\ge 0\), we write $$\theta_i=(\pi_i,\beta_i,\alpha_i,p_i,p_i',t_i,t_i',d_i,d_i')$$ for the deterministic parameter vector. 
The sequence \(\{\theta_i\}_{i\ge0}\) is part of the input to Algorithm~\ref{alg:triangle-free} and is fixed before the algorithm starts.
In each iteration, it constructs the round-\(i\) \LLL instance \(\mathfrak L_i\), invokes the \(\mathrm{SolveRoundLLL}\) subroutine in Algorithm~\ref{alg:solve-lll}, and obtains a realization \(X_i\) avoiding all events of \(\mathcal B_i\). Once \(X_i\) is fixed, the algorithm uses the resulting values of \(\mathrm{Sel}_i(u)\) and \(K_i(u)\) for the deterministic update. The stopping rule is \(d_{i-1}'>d_{\mathrm{stop}}\), where \(d_{\mathrm{stop}}=\Delta/(\log\Delta)^6\): as long as this holds, another iteration is performed; once it fails, Theorem~\ref{thm:clp20} from Section~\ref{sec:preliminaries} colors the current remaining graph \(G_{i-1}\) using \(d_{\mathrm{stop}}+1\) fresh colors.

\setcounter{algorithm}{2}
\begin{algorithm}[H]
\caption{IterativeColoringFramework$(G,\{\theta_i\}_{i\ge 0},r)$}
\label{alg:triangle-free}
\begin{algorithmic}[1]
\State Set $G_0=G$ and initialize the palettes $P_0(u)$ to be of $p_0$ colors for every vertex $u$
\State $d_{\mathrm{stop}}\gets \Delta/(\log\Delta)^6$
\State $i\gets 1$
\While{$d_{i-1}'>d_{\mathrm{stop}}$}
  \State construct the round-$i$ \LLL instance $\mathfrak L_i$
  \State $X_i\gets \mathrm{SolveRoundLLL}(i,r,\mathfrak L_i)$
  \ForAll{vertices $u\in V(G_{i-1})$}
    \State $\widetilde P_i(u)\gets K_i(u)\setminus \mathrm{Sel}_i(N_{i-1}(u))$
    \If{$\mathrm{Sel}_i(u)\cap \widetilde P_i(u)\neq\emptyset$}
      \State color $u$ with an arbitrary color in $\mathrm{Sel}_i(u)\cap \widetilde P_i(u)$
    \EndIf
  \EndFor
  \State let $U_i$ be the vertices of $G_{i-1}$ still uncolored after the round-$i$ coloring decisions
  \ForAll{vertices $u\in U_i$}
    \State $\widehat N_{i,c}(u)\gets \{v\in U_i\cap N_{i-1}(u): c\in \widetilde P_i(v)\}$ for every $c\in\widetilde P_i(u)$
    \State $\widehat P_i(u)\gets \{c\in \widetilde P_i(u): |\widehat N_{i,c}(u)|\le 2t_i\}$
    \If{$|\widehat P_i(u)|>p_i'$}
      \State set $P_i(u)$ to an arbitrary subset of $\widehat P_i(u)$ of size exactly $p_i'$
    \Else
      \State set $P_i(u)\gets \widehat P_i(u)$
    \EndIf
  \EndFor
  \State let $G_i$ be the graph induced by $U_i$
  \State $i\gets i+1$
\EndWhile
\State color the remaining graph $G_{i-1}$ using Theorem~\ref{thm:clp20} with $d_{\mathrm{stop}}+1$ fresh colors
\end{algorithmic}
\end{algorithm}

Recall from Section~\ref{sec:triangle-free-round-instance} that \(X_i(u)=(\mathrm{Sel}_i(u),K_i(u))\) for each active vertex \(u\), so \(X_i\) records all local round-\(i\) choices. After \(X_i\) has been fixed, the deterministic update in Algorithm~\ref{alg:triangle-free} first forms
\[
\widetilde P_i(u)=K_i(u)\setminus \mathrm{Sel}_i(N_{i-1}(u)),
\]
where \(\mathrm{Sel}_i(N_{i-1}(u))=\bigcup_{v\in N_{i-1}(u)}\mathrm{Sel}_i(v)\). A vertex is colored if one of its selected colors remains in this tentative palette. The set \(U_i\) consists of the vertices left uncolored after this decision. For \(u\in U_i\), the algorithm forms the color-specific neighborhoods \(\widehat N_{i,c}(u)=\{v\in U_i\cap N_{i-1}(u):c\in \widetilde P_i(v)\}\). By filtering, we mean deleting from \(\widetilde P_i(u)\) every color \(c\) with \(|\widehat N_{i,c}(u)|>2t_i\); the resulting palette is \(\widehat P_i(u)\). The algorithm then retains at most \(p_i'\) of these filtered colors. After this truncation, \(P_i(u)\) is the retained palette, while \(\widehat P_i(u)\) denotes the filtered palette before truncation. Also, \(N_{i,c}(u)\) denotes the residual neighbors of \(u\) whose retained palette still contains \(c\).

We now unfold the subroutine \(\mathrm{SolveRoundLLL}\) in Algorithm~\ref{alg:solve-lll}, which returns a realization \(X_i\) of the round-\(i\) variables avoiding all bad events in \(\mathcal B_i\). It takes the LLL instance \(\mathfrak L_i\) as input and constructs an \(r\)-partition via Theorem~\ref{thm:color-balanced-partition} with max-degree parameter \(\min\{d_{i-1}',\Delta\}\) and \(c\)-degree parameter \(2t_{i-1}\), and then runs the algorithm from Theorem~\ref{thm:davies12}. After this output \(X_i\) is fixed, the remaining update steps in Algorithm~\ref{alg:triangle-free} are deterministic.


\setcounter{algorithm}{1}
\begin{algorithm}[t]
\caption{$\mathrm{SolveRoundLLL}(i,r,\mathfrak L_i)$}
\label{alg:solve-lll}
\begin{algorithmic}[1]
\State Let $\mathfrak L_i=(\mathcal B_i,\mathcal V_i,\theta_i)$
\Statex \textbf{Output:} a realization of $\mathcal V_i$ in which no event of $\mathcal B_i$ occurs
\State compute an $r$-partition from Theorem~\ref{thm:color-balanced-partition} for the active variables in $\mathcal V_i$, with max-degree parameter $\min\{d_{i-1}',\Delta\}$ and $c$-degree parameter $2t_{i-1}$
\State run the algorithm from Theorem~\ref{thm:davies12} on $\mathfrak L_i$
\State \Return the resulting realization of $\mathcal V_i$
\end{algorithmic}
\end{algorithm}
\setcounter{algorithm}{3}

The next subsection specifies \(\mathcal B_i\) and the parameters so that avoiding all events in \(\mathcal B_i\) gives the palette and degree bounds needed for the next round.

\subsection{Bad Events, Parameters, and the Round Invariant}
\label{sec:triangle-free-parameters}

We now specify the concrete objects used in \(\mathfrak L_i\). For round \(i\ge1\), the prescribed parameter vector is \(\theta_i=(\pi_i,\beta_i,\alpha_i,p_i,p_i',t_i,t_i',d_i,d_i')\). The parameters \(\pi_i\) and \(\beta_i\) control sampling and color retention, while \(\alpha_i\) is the ideal attenuation factor. For \(p_i,t_i,d_i\) and their primed versions, we use the following convention throughout the section: the unprimed quantities are the ideal round-\(i\) values, and the primed quantities are analysis thresholds with slack. Thus \(p_i\) is the ideal retained palette size, whereas \(p_i'\) is the slightly smaller palette-size threshold used for truncation and for the padded averages; \(t_i\) is the ideal scale for the average \(c\)-degree, whereas \(t_i'\) is the slightly larger upper bound for the padded average \(c\)-degree in the invariant; and \(d_i\) is the ideal upper bound for residual degree, whereas \(d_i'\) is the slightly larger residual-degree bound promised after round \(i\).

We now fix the parameter schedule explicitly; these are the values of \(\theta_i\) supplied to Algorithm~\ref{alg:triangle-free} and to the round-\(i\) instance \(\mathfrak L_i\). Fix a small constant \(\eps>0\), choose a parameter \(k\) with \(1\le k\le(1-\eps)\ln\Delta/4\), set
\[
p_0=\frac{\Delta}{k} - \frac{\Delta}{(\log \Delta)^6}-1,
\]
Since \(k\le(1-\eps)\ln\Delta/4\), we have \(k/(\log\Delta)^6=o(1)\) and \(k/\Delta=o(1)\). Hence
\[
p_0=\frac{\Delta}{k}\left(1-\frac{k}{(\log\Delta)^6}-\frac{k}{\Delta}\right)
=(1-o(1))\frac{\Delta}{k}.
\]
We also define two exponent margins
\[
\eps_1=1-\frac{4k}{\ln\Delta}-\frac{2\eps}{3},
\qquad
\eps_2=1-\frac{4k}{\ln\Delta}-\frac{\eps}{3}.
\]
Thus \(\eps_1\ge\eps/3\) and \(\eps_2\ge2\eps/3\). The constant \(K\) sets the sampling scale, \(\delta\) sets the multiplicative slack between ideal and primed quantities, and \(T\) is the lower floor in the ideal average \(c\)-degree recurrence. We set
\[
K=\frac{4}{\eps},
\qquad
\delta=\frac{1}{\log^4\Delta},
\qquad
T=\Delta^{\eps_1/3}.
\]
Initialize \(p_0'=p_0\) and \(t_0=t_0'=d_0=d_0'=\Delta\). For \(i\ge1\), set
\[
\pi_i=\frac{1}{2Kt_{i-1}+1},
\]
and
\[
\beta_i=(1-\pi_i)^{2t_{i-1}},\qquad p_i=\beta_i p_{i-1},\qquad p_i'=(1-\delta/7)^i p_i,
\]
\[
\alpha_i=(1-\pi_i)^{(1-(1+2\delta)^{i-1}/2)p_i'}, \qquad
t_i=\max\{\alpha_i\beta_i t_{i-1},T\},\qquad
t_i'=(1+2\delta)^i t_i.
\]
For the degree upper bounds, we set
\[
d_{\mathrm{stop}}=\frac{\Delta}{(\log\Delta)^6},
\qquad
d_i=\max\!\left\{\alpha_i d_{i-1},\frac{d_{\mathrm{stop}}}{(1+\delta)^i}\right\},
\qquad
d_i'=(1+\delta)^i d_i
\]
for \(i\ge 1\). 

Given a round-\(i\) outcome, 
we use the surviving vertex set \(U_i\) and the tentative palette \(\widetilde P_i(u)\) produced by Algorithm~\ref{alg:triangle-free}. For each color \(c\in P_{i-1}(u)\), define
\[
\widetilde N_{i,c}(u)=\{x\in N_{i-1,c}(u): c\in \widetilde P_i(x)\}.
\]
Write
\[
\overline n_i(u)=\frac{1}{|P_i(u)|}\sum_{c\in P_i(u)} |N_{i,c}(u)|,
\qquad
\widetilde n_i(u)=\frac{1}{|\widetilde P_i(u)|}\sum_{c\in \widetilde P_i(u)} |N_{i,c}(u)|.
\]
If one of these denominators is zero, we interpret the corresponding average as \(0\).
We also set
\[
\lambda_i(u)=\min\left\{1,\frac{|P_i(u)|}{p_i'}\right\},
\qquad
\widetilde\lambda_i(u)=\min\left\{1,\frac{|\widetilde P_i(u)|}{p_i'}\right\},
\]
and
\[
D_i(u)=\lambda_i(u)\overline n_i(u)+(1-\lambda_i(u))2t_i,
\qquad
\widetilde D_i(u)=\widetilde\lambda_i(u)\widetilde n_i(u)+(1-\widetilde\lambda_i(u))2t_i.
\]
The round event \(F_i(u)\) is the statement that:
\begin{enumerate}[leftmargin=*,itemsep=0.25em]
  \item \(D_i(u)\le t_i'\),
  \item \(|N_i(u)|\le d_i'\).
\end{enumerate}
Let \(F_i\) be the event that \(F_i(u)\) holds for every vertex in \(U_i\).
We refer to \(F_i\) as the round-\(i\) invariant event.

The round-\(i\) bad-event family \(\mathcal B_i\) contains three types of local failures, with \(r=(\log\Delta)^{20}\) as in Theorem~\ref{thm:color-balanced-partition}:
\begin{align*}
&Q_i^+(u)
 := \left\{|\widetilde P_i(u)| < |P_{i-1}(u)|\beta_i\Bigl(1-\frac{\delta}{8}\Bigr)\Bigl(1-\frac{200}{r}\Bigr)\right\},\\
&R_i^+(u)
 := \left\{
  \sum_{c\in P_{i-1}(u)}
  \frac{|N_{i,c}(u)|}{|P_{i-1}(u)|}
  >
          \alpha_i\Bigl(1+\frac{\delta}{3}\Bigr)
          \left[
              \beta_i^2 \Bigl(\overline n_{i-1}(u)+\frac{\delta}{4}t_{i-1}\Bigr)+ \frac{300}{r}t_{i-1}  
  \right] + \frac{\delta}{3}T
\right\},\\
&M_i^+(u)
 := \left\{|N_i(u)|
  > \alpha_i\Bigl(1+\frac{\delta}{3}\Bigr)d_{i-1}'
    + \frac{200}{r}d_{i-1}'
    + \frac{\Delta}{r(\log\Delta)^8}\right\}.
\end{align*}
Here \(Q_i^+\) means that too few colors remain in the tentative palette,
\(R_i^+(u)\) means that the average residual \(c\)-degree is too large, and
\(M_i^+\) means that the residual degree is too large. When none of these
events occurs, the resulting palette and degree estimates imply \(F_i\);
this
implication is stated in Lemma~\ref{lem:triangle-free-round} below
and proved in Section~\ref{sec:triangle-free-lemma-proofs}.

Our Algorithm~\ref{alg:triangle-free} closely aligns with Variant~A of Algorithm~2 in Pettie and Su~\cite{PettieSu15}: 
it has the same iterative structure, and the parameter vector \(\theta_i\) follows the same schedule up to changes in a few constants. 
There are two substantive differences. 
First, the round instances are solved by our resilience \LLL subroutine in Algorithm~\ref{alg:solve-lll}, 
so we use modified bad events together with the corresponding entries of \(\theta_i\). 
Second, instead of switching to their later phase, we stop once the residual degree is low enough and add the events \(M_i^+(u)\) to certify that the final low-degree cleanup from Theorem~\ref{thm:clp20} can be applied. 
We also make the minor bookkeeping change of truncating palettes to simplify the analysis.

\subsection{Main Lemmas and Proof of Theorem~\ref{thm:triangle-free-main}}
\label{sec:triangle-free-main-lemmas}

We now state the three main lemmas needed to conclude the main triangle-free theorem. 
The first gives a fast solver for each local \LLL instance; the second says that eliminating the three bad-event families establishes \(F_i\) at the end of round \(i\); and the third says that \(d_i'\) reaches the final threshold quickly. Together with \(F_i\), this last statement is what makes the residual graph small enough for the final coloring step. Section~\ref{sec:triangle-free-lemma-proofs} proves these lemmas.

\begin{lemma}
With the parameter schedule fixed in Section~\ref{sec:triangle-free-parameters}, 
for every fixed \(C>0\), the following holds whenever \(\Delta\le(\log n)^C\). Suppose \(F_{i-1}\) holds and \(d_{i-1}'>d_{\mathrm{stop}}\). 
Then, with high probability, Algorithm~\ref{alg:solve-lll} solves the round-\(i\) instance \(\mathfrak L_i\) in \((\log\log n)^{O(1)}\) rounds.
\label{lem:triangle-free-instance}
\end{lemma}

\begin{lemma}
With the parameter schedule fixed in Section~\ref{sec:triangle-free-parameters}, assume \(F_{i-1}\) holds and \(d_{i-1}'>d_{\mathrm{stop}}\). For any realization \(X_i\) of the round-\(i\) variables in which no bad event of \(\mathcal B_i\) occurs, the deterministic post-processing step in Algorithm~\ref{alg:triangle-free} yields \(F_i\).
\label{lem:triangle-free-round}
\end{lemma}

\begin{lemma}
With the parameter schedule fixed in Section~\ref{sec:triangle-free-parameters}, \(d_i'\) reaches the final threshold: there is an integer \(\ell=O(k+\log\log\Delta)\) such that
\[
d_\ell'\le d_{\mathrm{stop}}=\Delta/(\log\Delta)^6.
\]
\label{lem:parameter-evolution}
\end{lemma}

\begin{proof}[Proof of Theorem~\ref{thm:triangle-free-main}]
Fix a constant \(\eps>0\) and a parameter \(k\) with \(1\le k\le(1-\eps)\ln\Delta/4\). Let \(\Delta_\eps\) be large enough so that all asymptotic inequalities used in the three supporting lemmas hold for every \(\Delta\ge\Delta_\eps\), and assume \(\Delta\ge\Delta_\eps\). We first dispose of the genuinely large-degree case. By Pettie and Su~\cite[Corollary~2]{PettieSu15}, there is a constant \(C=C(\eps)\) such that, whenever \(\Delta>(\log n)^C\), the desired \(\Delta/k\)-coloring is found in \(O(k+\log^* n)\) rounds. This is already \(O(k+(\log\log n)^{O(1)})\).

It remains to consider the range \(\Delta\le(\log n)^C\), where \(C\) is fixed.
 Let \(\ell\) be the first round with \(d_\ell'\le d_{\mathrm{stop}}\). Lemma~\ref{lem:parameter-evolution} gives \(\ell=O(k+\log\log\Delta)\), and since \(\Delta\le(\log n)^C\), this is \(O(k+(\log\log n)^{O(1)})\).
 It is enough to prove that \(F_i\) holds for every \(i\le \ell\) with high probability. Indeed, if \(F_\ell\) holds, then
\[
\Delta(G_\ell)=\max_{u\in U_\ell}|N_\ell(u)|\le d_\ell'\le d_{\mathrm{stop}},
\]
so the final residual graph is eligible for Theorem~\ref{thm:clp20} from Section~\ref{sec:preliminaries}.

We now prove that \(F_i\) holds throughout the process. Before the first iteration, \(F_0\) holds deterministically: \(G_0=G\), each initial palette has size \(p_0'=p_0\), every initial \(c\)-degree is at most \(\Delta=t_0'\), and \(|N_0(u)|\le \Delta=d_0'\) for every vertex \(u\).

 For the induction step, fix \(1\le i\le\ell\) and assume \(F_{i-1}\) holds at the start of iteration \(i\). By the minimality of \(\ell\), \(d_{i-1}'>d_{\mathrm{stop}}\). Algorithm~\ref{alg:triangle-free} constructs \(\mathfrak L_i\) and calls Algorithm~\ref{alg:solve-lll}. Lemma~\ref{lem:triangle-free-instance} shows that this call finishes in \((\log\log n)^{O(1)}\) rounds. Condition on the event that this solver call returns a realization \(X_i\) avoiding every bad event of \(\mathcal B_i\). Once \(X_i\) is fixed, Lemma~\ref{lem:triangle-free-round} implies that the deterministic update step establishes \(F_i\). Therefore \(F_i\) holds after every completed iteration up to \(\ell\), provided all invoked solvers have succeeded.

For success probability, Lemma~\ref{lem:triangle-free-instance} gives that each of the \(\ell\) calls to Algorithm~\ref{alg:solve-lll} succeeds with high probability. Since \(C\) is fixed, \(\ell=O(k+\log\log\Delta)=O(\log\log n)\) in this range. A union bound therefore implies that all \(\ell\) calls succeed simultaneously with probability at least \(1-n^{-c_1}\) for some constant \(c_1>1\), after decreasing the exponent constant if necessary.
Algorithm~\ref{alg:triangle-free} uses Theorem~\ref{thm:clp20} to color \(G_\ell\) in \(\tilde O(\log^{5/3}\log n)\) rounds, 
using \(d_{\mathrm{stop}}+1\) fresh colors, with probability at least \(1-n^{-c_2}\) for some constant \(c_2>1\). Taking one final union bound with the iteration-level error probabilities yields an overall success probability of at least \(1-n^{-c_3}\) for some constant \(c_3>1\), i.e. \(\mathrm{w.h.p.}\). 

The main phase uses \(p_0=\frac{\Delta}{k} - \frac{\Delta}{(\log \Delta)^6}-1\) colors, so the total number of colors is at most \(\Delta/k\). 

The total number of \LOCAL rounds is \(O(k+(\log\log n)^{O(1)})\): 
since \(\Delta\le(\log n)^C\), we have \(k= O(\log\log n)\) and hence \(\ell=O(k+\log\log\Delta)=O(\log\log n)\). 
Moreover, Lemma~\ref{lem:triangle-free-instance} shows that each call to Algorithm~\ref{alg:solve-lll} takes \((\log\log n)^{O(1)}\) rounds. 
Therefore the \(\ell\) per-iteration solver calls, together with the \(O(1)\) rounds of deterministic steps in each iteration, contribute only \((\log\log n)^{O(1)}\) rounds in total.
\end{proof}

\section{Proofs of the Triangle-Free Lemmas}
\label{sec:triangle-free-lemma-proofs}

This section provides the deferred proofs of the three lemmas from Section~\ref{sec:triangle-free}: 
Section~\ref{sec:triangle-free-numerical-proof} proves Lemma~\ref{lem:parameter-evolution}, 
Section~\ref{sec:triangle-free-instance-proof} proves Lemma~\ref{lem:triangle-free-instance}, 
Section~\ref{sec:triangle-free-resilience} proves the resilience propositions used in the instance proof, 
and Section~\ref{sec:triangle-free-round-proof} proves Lemma~\ref{lem:triangle-free-round}.

\subsection{Numerical Parameter Bounds}
\label{sec:triangle-free-numerical-proof}

We prove a more detailed numerical theorem, from which Lemma~\ref{lem:parameter-evolution} follows directly.

\begin{theorem}
\label{thm:executed-round-bounds}
Define the stopping time \(\tau_{\mathrm{end}}:=\inf\{i\ge 1:d_{i-1}'\le d_{\mathrm{stop}}\}\in \mathbb N\cup\{\infty\}\).
Assume \(r=(\log\Delta)^{20}\) and that \(\Delta\ge \Delta_\eps\). For the parameter choices fixed in Section~\ref{sec:triangle-free-parameters}, \(\tau_{\mathrm{end}}=O(k+\log\log\Delta)\). Moreover, for every \(i\) with \(1\le i<\tau_{\mathrm{end}}\),
\begin{align*}
\beta_i &\ge e^{-1/K},\\
p_{i-1}'\beta_i\Bigl(1-\frac{\delta}{8}\Bigr)\Bigl(1-\frac{200}{r}\Bigr) &\ge p_i',\\
\alpha_i\Bigl(1+\frac{\delta}{3}\Bigr)d_{i-1}'
  + \frac{200}{r}d_{i-1}'
  + \frac{\Delta}{r(\log\Delta)^8}
&\le d_i',\\
p_i' &\ge T(\log\Delta)^6.
\end{align*}
\end{theorem}

\begin{proof}
For every \(i\), the inequality \(\ln(1-x)\ge -x/(1-x)\) for \(0<x<1\), applied with \(x=1/(2Kt_{i-1}+1)\), gives
\[
\beta_i
=\left(1-\frac{1}{2Kt_{i-1}+1}\right)^{2t_{i-1}}
\ge \exp\!\left(-\frac{2t_{i-1}}{2Kt_{i-1}}\right)
\ge e^{-1/K}.
\]
Let \(s_i=t_i/p_i\).
Let \(\tau=\min\{i\ge 0:p_i\ge t_i\}\).

Set \(C:=\frac{1}{4Ke^{1/K}}\).
Fix \(i\le \min\{\tau,2k/C\}\). Since \(k\le(1-\eps)\ln\Delta/4\) and \(C\) is fixed, \(i=O(\log\Delta)\), so \(p_i'=(1-o(1))p_i\) and \((1+2\delta)^{i-1}=1+o(1)\). Also
\(\beta_i\ge e^{-1/K}\). Therefore
\[
\alpha_i
=\left(1-\frac{1}{2Kt_{i-1}+1}\right)^{(1-(1+2\delta)^{i-1}/2)p_i'}
\le \exp\!\left(-\frac{(1-o(1))p_i}{4Kt_{i-1}}\right)
\le \exp\!\left(-(1-o(1))C\frac{p_{i-1}}{t_{i-1}}\right).
\]
Hence, for every \(i<\min\{\tau,2k/C\}\),
\[
s_i=\frac{t_i}{p_i}=\alpha_i s_{i-1}
\le s_{i-1}\exp\!\left(-\frac{(1-o(1))C}{s_{i-1}}\right).
\]
Now let \(u:=k^{2/3}\). While \(s_{i-1}\ge u\), set \(w_i:=(1-o(1))C/s_{i-1}=o(1)\). Then
\[
s_i\le s_{i-1}e^{-w_i}\le s_{i-1}\bigl(1-(1-o(1))w_i\bigr)
\le s_{i-1}-(1-o(1))C.
\]
Thus, the number of steps needed to go from \(s_0=(1+o(1))k\) down to \(u\) is at most \(\frac{k-u}{(1-o(1))C}=(1+o(1))\frac{k}{C}\).
Once \(s_i\le u\), the recurrence is \(s_{i+1}\le s_i e^{-(1-o(1))C/s_i}\). For \(m\ge0\), whenever \(s_i\in[2^m,2^{m+1}]\), we have
\[
\frac{(1-o(1))C}{s_i}\ge \frac{(1-o(1))C}{2^{m+1}},
\]
hence each step decreases \(\log s_i\) by at least \((1-o(1))C/2^{m+1}\). Crossing one dyadic block changes \(\log s_i\) by at most \(\log 2\), so it takes at most \(O(2^m)\) steps. Summing over all blocks with \(2^m\le u\) gives
\[
\sum_{m=0}^{\lfloor \log_2 u\rfloor} O(2^m)=O(u).
\]
To see $\tau=O(\log \Delta)$, we argue by cases. If \(\tau<2k/C\), we are done since \(k\le(1-\eps)\ln\Delta/4\). 
Otherwise \(\tau\ge 2k/C\), so the recurrence is valid for every \(i<\min\{\tau,2k/C\}=2k/C\). In that case, the first phase takes at most \((1+o(1))k/C\) rounds and the second phase takes at most \(O(u)=o(k/C)\) rounds, so within the first \(2k/C\) rounds we obtain \(s_i\le1\), 
giving
\[
\tau\le (1+o(1))\frac{k}{C} = O(\log \Delta).
\]

For any fixed constant \(A\) and every \(j\in[\tau,\tau+A\log\log\Delta]\), we have \(j\le \tau+A\log\log\Delta\), hence
\[
p_j=\prod_{\ell=1}^{j}\beta_\ell\,p_0
\ge (1-o(1))e^{-j/K}\frac{\Delta}{k}
\ge (1-o(1))e^{-(\tau+A\log\log\Delta)/K}\frac{\Delta}{k}
\ge \Delta^{\,1-(1+o(1))\frac{k}{CK\ln\Delta}-o(1)}.
\]
Since \(C=\frac1{4Ke^{1/K}}\), this is
\[
 p_j\ge \Delta^{\,1-(1+o(1))4e^{1/K}\frac{k}{\ln\Delta}-o(1)}.
\]
Using \(K=4/\eps\) and \(k\le(1-\eps)\ln\Delta/4\),
\[
4\!\left(e^{1/K}-1\right)\frac{k}{\ln\Delta}
\le \left(e^{\eps/4}-1\right)(1-\eps)
\le \left(\frac{\eps}{4}+\frac{\eps^2}{16}\right)(1-\eps)
\le \frac{5\eps}{16}
<\frac{\eps}{3},
\]
where we used \(e^x-1\le x+x^2\) for \(x\in[0,1/4]\). 
Hence, for \(\Delta \ge \Delta_\eps\),
\begin{equation}\label{eq: larger than T}
   p_j\ge \Delta^{\,1-\frac{4k}{\ln\Delta}-\frac{\eps}{3}-o(1)}=\Delta^{\eps_2-o(1)} >\Delta^{\eps_1/3}=T.   
\end{equation}
By definition of \(\tau\), we have \(t_\tau\le p_\tau\). 
We now prove \(t_j\le p_j\) on this interval by induction on \(j\). 
The base case \(j=\tau\) is already noted. 
If \(\tau\le j<\tau+A\log\log\Delta\) and \(t_j\le p_j\), then
$\alpha_{j+1}\beta_{j+1}t_j\le \beta_{j+1}p_j=p_{j+1}$,
and, by \eqref{eq: larger than T} we also have \(t_{j+1}=\max\{\alpha_{j+1}\beta_{j+1}t_j,T\}\le p_{j+1}\).
Therefore \(t_j\le p_j\) for all \(j\in[\tau,\tau+A\log\log\Delta]\). In particular, for every \(\tau+1\le i\le\tau+A\log\log\Delta\),
\(\frac{p_i}{t_{i-1}}=\frac{\beta_i p_{i-1}}{t_{i-1}}\ge \beta_i\).
Now fix such an \(i\). Since \(i=O(\log\Delta)\), we obtain
\[
\alpha_i
\le \exp\!\left(-\frac{(1-o(1))p_i}{4Kt_{i-1}}\right)
\le \exp\!\left(-\frac{(1-o(1))e^{-1/K}}{4K}\right)
\]
and therefore \(\alpha_i\le \exp\!\left(-\frac{e^{-1/K}}{4K}\right)\). Hence, whenever \(\tau\le i\le\tau+A\log\log\Delta\) and \(d_i'>d_{\mathrm{stop}}\), we also have \(d_i>d_{\mathrm{stop}}/(1+\delta)^i\), so by the recurrence for \(d_i\),
\[
d_i=\alpha_i d_{i-1}
\le \exp\!\left(-\frac{e^{-1/K}}{4K}\right)d_{i-1}.
\]
Write \(c=\frac{e^{-1/K}}{4K}>0\). Iterating, for every \(s\ge 0\) with \(\tau+s<\tau_{\mathrm{end}}\),
\(d_{\tau+s}\le e^{-cs}d_\tau\). Also \(d_\tau\le \Delta\), so \(d_{\tau+s}\le e^{-cs}\Delta\).
Since \(\tau+s=O(\log\Delta)\) and \(\delta=(\log\Delta)^{-4}\),
\[
d_{\tau+s}'=(1+\delta)^{\tau+s}d_{\tau+s}
\le e^{\delta(\tau+s)}e^{-cs}\Delta
\le 2e^{-cs}\Delta,
\]
for \(\Delta\ge \Delta_\eps\). Therefore, if \(s=A\log\log\Delta\) and \(A>7/c\), then
\[
d_{\tau+s}'
\le 2\Delta\,e^{-cA\log\log\Delta}
=\frac{2\Delta}{(\log\Delta)^{cA}}
\le \frac{\Delta}{(\log\Delta)^6}
=d_{\mathrm{stop}}.
\]
Hence \(\tau_{\mathrm{end}}\le \tau+A\log\log\Delta+1\le (1+o(1))k/C+A\log\log\Delta+1=O(k+\log\log\Delta)\).

It remains to collect the bounds that hold in every executed round. For every \(1\le i<\tau_{\mathrm{end}}\), we have
\(p_i'\ge T(\log\Delta)^6\).
Indeed, \(i=O(k+\log\log\Delta)\) and \(p_i\ge \Delta^{\eps_2}\), so \(p_i'=(1-o(1))p_i\). Since \(\eps_2-\eps_1/3>0\), for  \(\Delta\ge \Delta_\eps\),
\[
p_i'\ge \Delta^{\eps_2}
\ge \Delta^{\eps_1/3}(\log\Delta)^6
=T(\log\Delta)^6.
\]

The palette inequality holds for every \(i\ge1\). Indeed, from \(p_i'=(1-\delta/7)^i p_i\), \(p_i=\beta_i p_{i-1}\), and \(200/r=o(\delta)\), we get
\[
p_{i-1}'\beta_i\Bigl(1-\frac{\delta}{8}\Bigr)\Bigl(1-\frac{200}{r}\Bigr)
\ge (1-\delta/7)\beta_i p_{i-1}'
=p_i'.
\]

For the degree-bound inequality, fix \(i\) with \(1\le i<\tau_{\mathrm{end}}\), and write
\[
E_i:=
\alpha_i\Bigl(1+\frac{\delta}{3}\Bigr)d_{i-1}'
  + \frac{200}{r}d_{i-1}'
  + \frac{\Delta}{r(\log\Delta)^8}.
\]
Since \(\alpha_j\le 1\) for every round \(j\), the recurrence for \(d_j'\) gives
\[
d_j'\le (1+\delta)^j\Delta
\qquad\text{for all }j\ge 0.
\]
Because \(i-1<\tau_{\mathrm{end}}=O(k+\log\log\Delta)=O(\log\Delta)\) and \(\delta=\log^{-4}\Delta\), we have \((1+\delta)^{i-1}\le 2\) for \(\Delta\ge \Delta_\eps\). Hence
\[
d_{i-1}'\le 2\Delta.
\]
Therefore
\[
\frac{200}{r}d_{i-1}'
  + \frac{\Delta}{r(\log\Delta)^8}
\le \frac{401\Delta}{(\log\Delta)^{20}}
\le \frac{\delta}{10}\,d_{\mathrm{stop}}
\]
since \(d_{\mathrm{stop}}=\Delta/(\log\Delta)^6\) and \(r=(\log\Delta)^{20}\).

If \(\alpha_i d_{i-1}\ge d_{\mathrm{stop}}/(1+\delta)^i\), then \(d_i=\alpha_i d_{i-1}\), hence
\[
d_i'=(1+\delta)^i d_i=(1+\delta)\alpha_i d_{i-1}'.
\]
The displayed lower bound on \(\alpha_i d_{i-1}'\) implies
\[
d_i'-\alpha_i\Bigl(1+\frac{\delta}{3}\Bigr)d_{i-1}'
=
\frac{2\delta}{3}\alpha_i d_{i-1}'
\ge \frac{2\delta}{3(1+\delta)}\,d_{\mathrm{stop}}
\ge \frac{\delta}{2}\,d_{\mathrm{stop}}.
\] 
Combining this with the bound on the additive error terms gives \(E_i\le d_i'\).

If \(\alpha_i d_{i-1}<d_{\mathrm{stop}}/(1+\delta)^i\), then \(d_i'=d_{\mathrm{stop}}\), while
\[
\alpha_i\Bigl(1+\frac{\delta}{3}\Bigr)d_{i-1}'
<
\frac{1+\delta/3}{1+\delta}\,d_{\mathrm{stop}}.
\]
Using the same additive-error estimate, we obtain
\[
E_i
<
\left(\frac{1+\delta/3}{1+\delta}+\frac{\delta}{10}\right)d_{\mathrm{stop}}
<
d_{\mathrm{stop}}
=
d_i'.
\]
Thus \(E_i\le d_i'\) for every \(i<\tau_{\mathrm{end}}\). This proves the theorem.
\end{proof}

\subsection{Proof of Lemma~\ref{lem:triangle-free-instance}}
\label{sec:triangle-free-instance-proof}

Fix a round \(i\) and assume that \(F_{i-1}\) holds. 
Since Lemma~\ref{lem:triangle-free-instance} is invoked only while Algorithm~\ref{alg:triangle-free} is still iterating, we also have \(d_{i-1}'>d_{\mathrm{stop}}\). 
As in Section~\ref{sec:triangle-free}, we work for sufficiently large \(\Delta\) (equivalently \(\Delta\ge \Delta_\eps\) after fixing \(\eps\)). 
Fix any partition \(\Phi=(\Phi_1,\dots,\Phi_r)\) that is an \(r\)-partition with parameters \(\min\{d_{i-1}',\Delta\}\) and \(2t_{i-1}\) in the sense of Definition~\ref{def:r-partition-degree-params} from Section~\ref{sec:color-balanced-partition}. 
The three propositions below, proved in Section~\ref{sec:triangle-free-resilience}, establish resilience for all bad events in \(\mathcal B_i\) with respect to \(\Phi\). 
Once these three statements are established for the same partition \(\Phi\), Theorem~\ref{thm:davies12} can be applied to the whole LLL instance \(\mathfrak L_i\).

\begin{proposition}
\label{prop:q-resilient}
Assume \(F_{i-1}\) holds and \(d_{i-1}'>d_{\mathrm{stop}}\). 
Then for every active vertex \(u\), 
\[
\Prob_{\mathcal V_i^1}\left[
\bigcup_{j\le r}\bigcup_{S\subseteq \Phi_j}
\left\{
\Prob_{\mathcal V_i^2}\left[(Q_i^+(u))^S \mid \mathcal V_i^1\right]
> 2d^{-20}
\right\}
\right]
\le d^{-35}.
\]
In particular, 
the bad event \(Q_i^+(u)\) is resilient with respect to \(\Phi\).
\end{proposition}

\begin{proposition}
\label{prop:r-resilient}
Assume \(F_{i-1}\) holds and \(d_{i-1}'>d_{\mathrm{stop}}\). Then for every active vertex \(u\), the bad event \(R_i^+(u)\) is resilient with respect to \(\Phi\).
\end{proposition}

\begin{proposition}
\label{prop:m-resilient}
Assume \(F_{i-1}\) holds and \(d_{i-1}'>d_{\mathrm{stop}}\). Then for every active vertex \(u\), the bad event \(M_i^+(u)\) is resilient with respect to \(\Phi\).
\end{proposition}

In Proposition~\ref{prop:q-resilient}, we prove a slightly stronger notion of resilience of \(Q_i^+(u)\) to support the proof of Proposition~\ref{prop:r-resilient}.
Given these propositions, we can prove Lemma~\ref{lem:triangle-free-instance} as below.

\begin{proof}[Proof of Lemma~\ref{lem:triangle-free-instance}]
Fix \(C>0\) such that \(\Delta\le(\log n)^C\), and let \(b\) be the constant supplied by Theorem~\ref{thm:color-balanced-partition} from Section~\ref{sec:color-balanced-partition} with \(a=C\). 
The assumptions of that theorem hold as follows.
For every active vertex \(u\), the invariant \(F_{i-1}\) gives \(|N_{i-1}(u)|\le d_{i-1}'\), and the original maximum-degree bound gives \(|N_{i-1}(u)|\le\Delta\), so
\[
|N_{i-1}(u)|\le \min\{d_{i-1}',\Delta\}.
\]
For the \(c\)-degree condition, every \(c\in P_{i-1}(u)\) satisfies \(|N_{i-1,c}(u)|\le 2t_{i-1}\).
It remains to check the lower bounds on \(\min\{d_{i-1}',\Delta\}\) and \(2t_{i-1}\) required by Theorem~\ref{thm:color-balanced-partition}. Since this iteration is executed only while \(d_{i-1}'>d_{\mathrm{stop}}\), we have
\[
\min\{d_{i-1}',\Delta\}\ge d_{\mathrm{stop}}=\Delta/(\log\Delta)^6\ge(\log\Delta)^b,
\]
for \(\Delta\ge \Delta_\eps\). 
Also, by  definition of \(t_j\), we have 
\[2t_{i-1}\ge 2T=2\Delta^{\eps_1/3} \ge(\log\Delta)^b.
\] 
Thus Theorem~\ref{thm:color-balanced-partition} applies to \(H=G_{i-1}\), \(U=V(G_{i-1})\), and the palettes \(P_{i-1}\), with max-degree parameter \(\min\{d_{i-1}',\Delta\}\) and \(c\)-degree parameter \(2t_{i-1}\). 
Therefore, Algorithm~\ref{alg:solve-lll} constructs with high probability an \(r\)-partition \(\Phi=(\Phi_1,\dots,\Phi_r)\) that satisfies Definition~\ref{def:r-partition-degree-params} for these parameters. 
Condition on this event and fix that \(\Phi\).

By Propositions~\ref{prop:q-resilient}, \ref{prop:r-resilient}, and~\ref{prop:m-resilient}, every bad event in \(\mathcal B_i\) is resilient with respect to \(\Phi\). 
Therefore Theorem~\ref{thm:davies12} applies to the round instance with the partition \(\Phi\). Because \(r=(\log\Delta)^{20}\le(\log\log n)^{O(1)}\) in the present range, the partition construction and Theorem~\ref{thm:davies12} together take \((\log\log n)^{O(1)}\) rounds with high probability. This proves the claimed guarantee.
\end{proof}

\subsection{Event-by-Event Resilience}
\label{sec:triangle-free-resilience}

\subsubsection*{Shared Setup}

Throughout this subsection, fix the round \(i\), assume \(F_{i-1}\), and let
\(\Phi=(\Phi_1,\dots,\Phi_r)\) be the \(r\)-partition supplied by the proof of
Lemma~\ref{lem:triangle-free-instance}. We will use its two consequences
separately. For every active vertex \(u\), every color \(c\in P_{i-1}(u)\),
and every part \(\Phi_j\),
\[
|N_{i-1}(u)\cap \Phi_j|
\le
100\min\{d_{i-1}',\Delta\}/r
\le
100d_{i-1}'/r,
\]
and
\[
|N_{i-1,c}(u)\cap \Phi_j|
\le
200t_{i-1}/r.
\]
The first estimate is the lightness bound used in the max-degree argument, and
the second is the color-balance bound used in the palette and \(c\)-degree
arguments.

To verify resilience, we use the mixed-sample notation from
Section~\ref{sec:preliminaries}. Let \(\mathcal V_i^1\) and \(\mathcal V_i^2\)
be independent copies of the round-\(i\) variable family \(\mathcal V_i\). For
any event \(A\) determined by the round-\(i\) variables and any set
\(S\subseteq \Phi_j\), let \(A^S\) denote the mixed-sample event obtained by taking
the variables in \(S\) from \(\mathcal V_i^2\) and all remaining variables
from \(\mathcal V_i^1\). If \(Z_i\) is any round-\(i\) object determined by
the sampled variables, such as \(\mathrm{Sel}_i(u)\), \(K_i(u)\),
\(\widetilde P_i(u)\), \(\widetilde N_{i,c}(u)\), or \(N_i(u)\), then
\(Z_i^\ell\) denotes the corresponding object computed from the sample
\(\mathcal V_i^\ell\), and \(Z_i^S\) denotes the corresponding object computed
from the mixed sample that uses \(\mathcal V_i^2\) on \(S\) and
\(\mathcal V_i^1\) outside \(S\). In the proofs below, the target events are
indexed either by a vertex \(u\) or by a pair \((u,c)\), and only vertices in
\(N_{i-1}(u)\cap \Phi_j\) or \(N_{i-1,c}(u)\cap \Phi_j\) can change the indexed
quantity under one-part resampling.

We repeatedly invoke Theorem~\ref{thm:executed-round-bounds}. Since
\(d_{i-1}'>d_{\mathrm{stop}}\), the round \(i\) is executed before the stopping
time, so that theorem gives
\[
\beta_i\ge e^{-1/K},
\qquad
p_{i-1}'\beta_i\Bigl(1-\frac{\delta}{8}\Bigr)\Bigl(1-\frac{200}{r}\Bigr)\ge p_i',
\]
\[
\alpha_i\Bigl(1+\frac{\delta}{3}\Bigr)d_{i-1}'
  + \frac{200}{r}d_{i-1}'
  + \frac{\Delta}{r(\log\Delta)^8}
\le d_i',
\qquad
p_i'\ge T(\log\Delta)^6,
\]
and \(i=O(k+\log\log\Delta)=O(\log\Delta)\). In particular,
\(t_{i-1}'=(1+2\delta)^{i-1}t_{i-1}\le 2t_{i-1}\) for \(\Delta\ge \Delta_\eps\), and every concentration estimate below is applied in a polylog-large
range.

We also introduce the following bad events.
\begin{align*}
Q_i(u)
&:=
\left\{
|\widetilde P_i(u)|
<
|P_{i-1}(u)|\beta_i\Bigl(1-\frac{\delta}{8}\Bigr)
\right\},\\
M_i(u)
&:=
\left\{|N_i(u)|
  > \alpha_i\Bigl(1+\frac{\delta}{5}\Bigr)d_{i-1}'\right\}.
\end{align*}
Thus \(Q_i^+(u)\subseteq Q_i(u)\) and \(M_i^+(u)\subseteq M_i(u)\).
The following simple fact will be useful in the next proofs.

\begin{lemma}
\label{fact:palette-lower}
If \(F_{i-1}(u)\) holds, then
\[
|P_{i-1}(u)|\ge \left(1-\frac{(1+2\delta)^{i-1}}{2}\right)p_{i-1}'.
\]
In particular, \(|P_{i-1}(u)|\ge \left(\frac12-o(1)\right)p_{i-1}\).
\end{lemma}

\begin{proof}
From \(F_{i-1}(u)\), we have \((1-\lambda_{i-1}(u))2t_{i-1}\le D_{i-1}(u)\le t_{i-1}'=(1+2\delta)^{i-1}t_{i-1}\), so \(\lambda_{i-1}(u)\ge 1-\frac{(1+2\delta)^{i-1}}{2}\), which gives the first bound.
Also \(p_{i-1}'=(1-\delta/7)^{i-1}p_{i-1}\). Since \(i<\tau_{\mathrm{end}}\), Theorem~\ref{thm:executed-round-bounds} from Section~\ref{sec:triangle-free-numerical-proof} gives \(i=O(k+\log\log\Delta)=O(\log\Delta)\), and with \(\delta=1/\log^4\Delta\) we obtain \((1+2\delta)^{i-1}=1+o(1)\) and \((1-\delta/7)^{i-1}=1-o(1)\). Therefore, \(|P_{i-1}(u)|\ge \left(\frac12-o(1)\right)p_{i-1}\).
\end{proof}

\begin{proof}[Proof of Proposition~\ref{prop:q-resilient}]
Fix a part \(\Phi_j\), and a subset \(S\subseteq \Phi_j\). 
Define \(\widetilde Q_i(u)\) in \(\mathcal V_i^1\) as follows:
\[
\widetilde Q_i(u)
:=
Q_i(u)
\cup
\left\{
\Prob_{\mathcal V_i^2}\left[
\bigl(Q_i(u)\bigr)^{\{u\}} \middle| \mathcal{V}_i^1
\right]
> d^{-20}
\right\}.
\]
The second term is included to handle the case in the resilience definition where the adversarial set \(S\) contains the variable \(X_i(u)\). 
By the notation of Section~\ref{sec:resilient-lll}, \(\bigl(Q_i(u)\bigr)^{\{u\}}\) is the event \(Q_i(u)\) evaluated with \(X_i(u)\) taken from the second sample and all other round-\(i\) variables taken from the first sample. 
Thus \(\widetilde Q_i(u)\) rules out, in advance, the bad first samples for which merely resampling \(K_i(u)\) would make the palette event likely.
By Lemma~\ref{lem:appendix-degree-tail} from
Appendix~\ref{sec:appendix-degree-tail},
\[
\Prob[Q_i(u)]\le \Delta^{-300} \le d^{-75}.
\]
Also, since $\Prob_{\mathcal{V}_i^1}[\cdot]$ and $\Prob_{\mathcal{V}_i^2}[\cdot]$ are the same distribution, this implies that
\[
\E_{\mathcal V_i^1}\left[
\Prob_{\mathcal V_i^2}\left[\bigl(Q_i(u)\bigr)^{\{u\}} \mid \mathcal{V}_i^1  \right]
\right]
=
\Prob[Q_i(u)]
\le
d^{-75},
\]
so by Markov's inequality
\[
\Prob_{\mathcal{V}_i^1}\left[
\Prob_{\mathcal V_i^2}\left[\bigl(Q_i(u)\bigr)^{\{u\}} \mid \mathcal{V}_i^1  \right]>d^{-20}
\right]
\le
d^{-55}.
\]
Hence \(\Prob_{\mathcal{V}_i^1}[\widetilde Q_i(u)]\le d^{-35}\).

Assume that \(\widetilde Q_i(u)\) does not occur for the variables in \(\mathcal V_i^1\). In particular, \(Q_i(u)\) does not occur.
If \(u\notin S\), then the set \(K_i(u)\) is determined by \(\mathcal V_i^1\), 
and the only additional palette loss due to resampling comes from vertices of \(S\cap N_{i-1}(u)\). For each color \(c\in \widetilde P_i^{1}(u)\), the event that \(c\) survives the resampling inside \(S\) depends only on the \(c\)-selection bits of the vertices in \(S\cap N_{i-1,c}(u)\). These events are independent across colors. Moreover, by the partition bound \(|S\cap N_{i-1,c}(u)|\le |\Phi_j\cap N_{i-1,c}(u)|\le 200t_{i-1}/r\) and the definition \(\pi_i=1/(2Kt_{i-1}+1)\),
\begin{align*}
\Prob_{\mathcal V_i^2}\left[
  c\notin \mathrm{Sel}_i^{2}(S\cap N_{i-1,c}(u))
  \mid \mathcal V_i^1
\right]
&=
(1-\pi_i)^{|S\cap N_{i-1,c}(u)|} \\
&\ge
\left(1-\frac{1}{2Kt_{i-1}+1}\right)^{200t_{i-1}/r} \\
&\ge
1-\frac{200t_{i-1}}{r(2Kt_{i-1}+1)}.
\end{align*}
Then \(|\widetilde P_i^S(u)|\) is a sum of independent indicators indexed by the colors in \(\widetilde P_i^1(u)\), 
and the calculation above shows that each indicator has conditional mean at least \(1-\frac{200t_{i-1}}{r(2Kt_{i-1}+1)}\).
Hence
\[
\mu:=\E_{\mathcal V_i^2}\left[|\widetilde P_i^S(u)|\mid \mathcal V_i^1,\ \neg\widetilde Q_i(u)\right]
\ge
|P_{i-1}(u)|\beta_i\Bigl(1-\frac{\delta}{8}\Bigr)\left(1-\frac{200t_{i-1}}{r(2Kt_{i-1}+1)}\right).
\]
By Lemma~\ref{fact:palette-lower}, \(|P_{i-1}(u)|\ge \left(\frac12-o(1)\right)p_{i-1}\). Also \(\frac{200t_{i-1}}{2Kt_{i-1}+1}\le \frac{100}{K}\le 25\), since \(K=4/\eps\ge 4\). Hence
\[
\mu
\ge
\left(\frac12-o(1)\right)e^{-1/K}\left(1-\frac{\delta}{8}\right)\left(1-\frac{25}{r}\right)p_{i-1}
=
\Omega(p_{i-1}).
\]
Hence \(\mu=\Omega(p_i)\) as well, since \(p_i=\beta_i p_{i-1}\le p_{i-1}\).
The bad event \((Q_i^+(u))^S\) is precisely
\[
|\widetilde P_i^S(u)|<
|P_{i-1}(u)|\beta_i\Bigl(1-\frac{\delta}{8}\Bigr)\Bigl(1-\frac{200}{r}\Bigr)
\le
\left(1-\frac{100}{r}\right)\mu.
\]
Since \(|\widetilde P_i^S(u)|\) is a sum of independent Bernoulli variables with mean \(\mu\), Chernoff's bound gives
\[
\Prob_{\mathcal V_i^2}\left[
|\widetilde P_i^S(u)|<\left(1-\frac{100}{r}\right)\mu
\ \middle|\ \mathcal V_i^1,\ \neg \widetilde Q_i(u)
\right]
\le
\exp\!\left(-\frac{10000\mu}{2r^2}\right)
=
\exp\!\left(-\frac{5000\,\mu}{r^2}\right).
\]
Finally, Theorem~\ref{thm:executed-round-bounds} gives \(p_{i-1}'\ge T(\log\Delta)^6\), and \(p_{i-1}'=(1-\delta/7)^{i-1}p_{i-1}\) with \(i=O(\log\Delta)\) implies \(p_{i-1}=(1+o(1))p_{i-1}'\). 
Since \(\mu=\Omega(p_{i-1})\), we obtain \(\exp\!\left(-\frac{5000\,\mu}{r^2}\right)\le d^{-20}\).

If \(u\in S\), then \(K_i(u)\) is also resampled. 
This is the only additional difficulty beyond the \(u\notin S\) case: resampling \(X_i(u)\) can change \(K_i(u)\) itself. 
On the complement of the second term in \(\widetilde Q_i(u)\), 
the palette \(\widetilde P_i^{\{u\}}(u)\) of $u$ after resampling $K_i(u)$ has size at least \(|P_{i-1}(u)|\beta_i(1-\delta/8)\) except with probability \(d^{-20}\). Condition on this good event. After that conditioning, the only remaining loss again comes from the vertices in \(S\cap N_{i-1}(u)\), so the same Chernoff calculation gives the stronger bound \(d^{-20}\) for the extra loss. Therefore
\[
\Prob_{\mathcal V_i^2}\left[(Q_i^+(u))^S \middle| \mathcal V_i^1,\ \neg \widetilde Q_i(u)\right]\le 2d^{-20}.
\]
For each \(j\), define
\[
E_j:=\bigcup_{S\subseteq \Phi_j}\left\{\Prob_{\mathcal V_i^2}\left[(Q_i^+(u))^S \middle| \mathcal V_i^1\right]>2d^{-20}\right\}.
\]
Crucially, no union bound over the exponentially many subsets \(S\subseteq\Phi_j\) is needed. 
The preceding argument was uniform in \(S\): after \(\widetilde Q_i(u)\) is excluded, for every possible $S$, the bad event under resampling has conditional probability at most \(2d^{-20}\).
By the bound above, if \(\widetilde Q_i(u)\) does not occur then \(E_j\) does not occur either, so \(E_j\subseteq \widetilde Q_i(u)\). The same inclusion holds for every \(j\) because the argument is uniform in \(j\). Therefore
\(\bigcup_{j\le r}E_j\subseteq \widetilde Q_i(u)\), so \(\Prob_{\mathcal V_i^1}\left[\bigcup_{j\le r}E_j\right]\le \Prob_{\mathcal V_i^1}[\widetilde Q_i(u)]\le d^{-35}\).
This proves the stronger statement in the proposition, and hence \(Q_i^+(u)\) is resilient.
\end{proof}

\begin{proof}
[Proof of Proposition~\ref{prop:r-resilient}]
Fix a vertex \(u\), and a part \(\Phi_j\) from $\Phi$.
Let
\[
H_i^+(u)
:=
\left\{\sum_{c\in P_{i-1}(u)} |\widetilde N_{i,c}(u)|
  > \beta_i^2 |P_{i-1}(u)|\Bigl(\overline n_{i-1}(u)+\frac{\delta}{4}t_{i-1}\Bigr)
    + \frac{300}{r}t_{i-1}|P_{i-1}(u)|\right\}.
\]
We will prove the following two claims.
\begin{enumerate}
  \item With probability at least \(1-2d^{-40}\) over the choice of the first sample \(\mathcal V_i^1\), 
  \begin{equation}\label{eq:H plus 0}
  \Prob_{\mathcal V_i^2}\left[(H_i^+(u))^S \mid \mathcal V_i^1\right]=0
  \qquad\text{for every }S\subseteq \Phi_j.
  \end{equation}
  \item For each $c\in P_{i-1}(u)$, with probability at least \(1-2d^{-34}\) over the choice of the first sample \(\mathcal V_i^1\), for every \(S\subseteq \Phi_j\) we have
  \[
  \Prob_{\mathcal V_i^2}\left[
  |N_{i,c}^S(u)| >  |\widetilde N_{i,c}^S(u)|\alpha_i\bigl(1+\frac{\delta}{3}\bigr)+\frac{\delta}{3}T \mid \mathcal V_i^1
  \right] < 3d^{-18},
  \]
\end{enumerate}
where we write \(\widetilde N_{i,c}^S(u)\) and $N_{i,c}^S(u)$ for the mixed-sample version of \(\widetilde N_{i,c}(u)\) and $N_{i,c}(u)$ respectively.
We will prove the two claims below. 
For each fixed \(j\), except on a first-sample event of probability at most \(2d^{-40}+2|P_{i-1}(u)|d^{-34}\), Claim~1 holds and Claim~2 holds simultaneously for all \(c\in P_{i-1}(u)\) by a union bound. 
Then, for every \(S\subseteq \Phi_j\), a union bound over colors gives
\[
\Prob_{\mathcal V_i^2}\left[(R_i^+(u))^S \mid \mathcal V_i^1\right]
\le 3|P_{i-1}(u)|d^{-18}\le d^{-10},
\]
since on the complement of the bad events from Claim~2 and Claim~1,  \((R_i^+(u))^S\) does not occur. 
Taking a final union bound over \(j=1,\ldots,r\),  
we obtain the resilience of \(R_i^+(u)\) with respect to \(\Phi\).

To prove the two claims, we organize the random variables $\{X_i(v)\}$ in the set $S$ into two parts. 
Let
\[
S_a:= \bigcup_{v\in N_{i-1}(u)\cap S} K_i(v) \cup \left(\bigcup_{v\in N_{i-1}(u)}\bigcup_{w\in N_{i-1}(v)\cap S} \mathrm{Sel}_i(w)\right),
\qquad
S_b:= \bigcup_{v\in N_{i-1}(u)\cap S} \mathrm{Sel}_i(v).
\]
Note that these are the only variables that can affect the event $R_i^+(u)$, and that $S_a$ and $S_b$ are disjoint since $G_{i-1}$ is triangle-free.
Moreover, for the resilience of $H_i^+(u)$, we only need to consider the variables in $S_a$. 
In particular,
$\widetilde N_{i,c}^S(u) = \widetilde N_{i,c}^{S_a}(u) $ and thus
$(H^+_i(u))^{S}=(H^+_i(u))^{S_a}$.

For the first claim, we first define an auxiliary event 
\[
H_i(u)
:=
\left\{\sum_{c\in P_{i-1}(u)} |\widetilde N_{i,c}(u)|
  > \beta_i^2 |P_{i-1}(u)|\Bigl(\overline n_{i-1}(u)+\frac{\delta}{4}t_{i-1}\Bigr)\right\},
\]
and note that \(H_i^+(u)\subseteq H_i(u)\).
By Lemma~\ref{lem:appendix-degree-tail} from Appendix~\ref{sec:appendix-degree-tail},
$
\Prob_{\mathcal{V}_i^1}[H_i(u)]\le \Delta^{-200}\le d^{-40}
$.
We will control how much the resampling inside \(\Phi_j\) can increase
\[
\sum_{c\in P_{i-1}(u)} |\widetilde N_{i,c}(u)|.
\]
Fix \(c\in P_{i-1}(u)\) and a vertex \(x\in N_{i-1,c}(u)\). 
If \(x\) belongs to \(\widetilde N_{i,c}^S(u)\) but not to \(\widetilde N_{i,c}(u)\), 
then this change must be caused by the resampling inside \(\Phi_j\). 
There are two possibilities. 
First, \(x\) itself may lie in \(\Phi_j\). 
Second, in the first sample some vertex of \(\Phi_j\) selected the color \(c\) and thereby prevented \(x\) from being counted, whereas after resampling no vertex of \(\Phi_j\) selects \(c\). 

The vertices of the first type contribute at most
\[
\sum_{c\in P_{i-1}(u)} |N_{i-1,c}(u)\cap \Phi_j|
\le
\frac{200}{r}t_{i-1}|P_{i-1}(u)|.
\]
To capture the second type, define
\[
A_{i,j}(u,c)
:=
\bigl\{
x\in N_{i-1,c}(u)\setminus \Phi_j:
 c\in \mathrm{Sel}_i^{1}(y)\text{ for some }y\in \Phi_j\cap N_{i-1,c}(x)
\bigr\}.
\]
If \(x\in (\widetilde N_{i,c}^S(u)\setminus \widetilde N_{i,c}(u))\setminus \Phi_j\), then \(x\) must arise from the second type above, and hence \(x\in A_{i,j}(u,c)\). Therefore
\[
\bigl(\widetilde N_{i,c}^S(u)\setminus \widetilde N_{i,c}(u)\bigr)\setminus \Phi_j
\subseteq A_{i,j}(u,c)
\]
for every \(S\subseteq \Phi_j\). 
For a fixed vertex \(x\in N_{i-1,c}(u)\setminus \Phi_j\), 
we have \(|\Phi_j\cap N_{i-1,c}(x)|\le 200t_{i-1}/r\), so
\[
\Prob_{\mathcal V_i^1}[x\in A_{i,j}(u,c)]
\le \frac{200\pi_i t_{i-1}}{r}
= \frac{200t_{i-1}}{r(2Kt_{i-1}+1)}
\le \frac{100}{Kr}.
\]
Summing over \(x\) and then over colors yields
\begin{align*}
\mu
:=\E\left[\sum_{c\in P_{i-1}(u)} |A_{i,j}(u,c)|\right]
&\le \frac{100}{Kr}\sum_{c\in P_{i-1}(u)} |N_{i-1,c}(u)| \\
&\le \frac{200}{Kr}\,t_{i-1}|P_{i-1}(u)|\\
&\le \frac{50}{r}\,t_{i-1}|P_{i-1}(u)|,
\qquad\text{since }K=4/\eps\ge 4.
\end{align*}
Let
\[
B_{i,j}(u):=
\left\{\sum_{c\in P_{i-1}(u)} |A_{i,j}(u,c)| > \frac{100}{r}t_{i-1}|P_{i-1}(u)|\right\}.
\]
The random variables \(|A_{i,j}(u,c)|\) are independent across colors, and each is bounded by \(|N_{i-1,c}(u)|\le 2t_{i-1}\), so Hoeffding's inequality gives
\begin{align*}
\Prob[B_{i,j}(u)]
&\le \exp\!\left(
-\frac{2\left(\frac{100}{r}t_{i-1}|P_{i-1}(u)|-\mu\right)^2}{4t_{i-1}^2|P_{i-1}(u)|}
\right) \\
&\le \exp\!\left(
-\frac{2\left(\frac{50}{r}t_{i-1}|P_{i-1}(u)|\right)^2}{4t_{i-1}^2|P_{i-1}(u)|}
\right)
=
\exp\!\left(-\frac{1250\,|P_{i-1}(u)|}{r^2}\right).
\end{align*}
By Lemma~\ref{fact:palette-lower}, \(|P_{i-1}(u)|/r^2\gg \log d\), so this is at most \(d^{-40}\).
Then \(\Prob_{\mathcal{V}_i^1}[B_{i,j}(u)]\le d^{-40}\), and on the complement of \(B_{i,j}(u)\) we have
\[
\sum_{c\in P_{i-1}(u)} |\widetilde N_{i,c}^S(u)|
\le
\sum_{c\in P_{i-1}(u)} |\widetilde N_{i,c}(u)|
\,+\,
\frac{300}{r}t_{i-1}|P_{i-1}(u)|
\]
for every \(S\subseteq \Phi_j\). Therefore, 
if we fix a realization of \(\mathcal V_i^1\) for which neither \(H_i(u)\) nor \(B_{i,j}(u)\) occurs, then the defining inequality of \(H_i^+(u)\) cannot hold in the mixed sample. Thus, for such $\mathcal V_i^1$,
\[
\Prob_{\mathcal V_i^2}\left[(H_i^+(u))^S \mid \mathcal{V}_i^1 \right]=0
\qquad\text{for every }S\subseteq \Phi_j.
\]

Next we will work on the second claim. 
Note that the events $\{Q^+_i(v)\}_{v\in N_{i-1}(u)}$ depend only on 
\[V^{1,a}_i:= \bigcup_{v\in N_{i-1}(u)}  \left(K_i^1(v) \cup\bigcup_{w\in N_{i-1}(v)} \mathrm{Sel}^1_i(w)\right),\] 
so $\bigl(\bigcup_{v\in N_{i-1}(u)} Q^+_i(v)\bigr)^S = \bigl(\bigcup_{v\in N_{i-1}(u)} Q^+_i(v)\bigr)^{S_a}$. 
By Proposition~\ref{prop:q-resilient} and a union bound over $v\in N_{i-1}(u)$, 
with probability at least \(1-d^{-34}\) over the choice of the first sample \(\mathcal V_i^1\), we have 
\begin{equation}\label{eq:Q plus bound}
\Prob_{\mathcal{V}_i^2}\left[\Bigl(\bigcup_{v\in N_{i-1}(u)} Q^+_i(v)\Bigr)^{S_a} \mid  \mathcal{V}_i^{1,a}\right]\le d^{-18}   \qquad\text{for every }S\subseteq \Phi_j.
\end{equation}

Fix a realization of $\mathcal{V}_i^{1,a}$ such that \eqref{eq:Q plus bound} holds.
For a given $S\subseteq \Phi_j$, assume that the assignment in $V_2^i$ on $S_a$ is such that $\bigl(\bigcup_{v\in N_{i-1}(u)} Q^+_i(v)\bigr)^{S_a}$ does not occur. 
We refer to this event as $\mathcal{E}(S_a)$.
Then under $\mathcal{E}(S_a)$, for each $v\in N_{i-1}(u)$, we have $|\widetilde P_i^{S_a}(v)|\ge  |P_{i-1}(v)|\beta_i(1-\delta/8)$. 
Since
\[
  |N_{i,c}^S(u)| \le |N_{i,c}^S(u) \cap S| + |N_{i,c}^S(u) \setminus S|,
\]
it suffices to derive upper bounds on both $|N_{i-1,c}(u) \cap S|$ and $|N_{i,c}^S(u) \setminus S|$.
Observe that a node $v\in\widetilde{N}_{i,c}^{S_a}(u)$ belongs to $N_{i,c}^S(u)$ if and only if $\mathrm{Sel}^{S_b}_i(v) \cap \widetilde{P}_i^{S_a}(v) = \emptyset$.
If $v\in S$, then $v\in N_{i,c}^S(u)$ if and only if $\mathrm{Sel}^{2}_i(v) \cap \widetilde{P}_i^{S_a}(v) = \emptyset$;
otherwise $v\in N_{i,c}^S(u)$ if and only if $\mathrm{Sel}^{1}_i(v) \cap \widetilde{P}_i^{S_a}(v) = \emptyset$.
The random variables $\bigcup_{v\in N_{i-1}(u)} \bigl( \mathrm{Sel}^1_i(v) \bigr)=:V_i^{1,b}$ are independent of $V_i^{1,a}$ and the event $\mathcal{E}(S_a)$.
Under $\mathcal{E}(S_a)$, by Lemma~\ref{fact:palette-lower}, the event $\mathrm{Sel}^{1}_i(v) \cap \widetilde{P}_i^{S_a}(v) = \emptyset$ happens with probability at most 
\begin{equation}
  \label{eq: ub by alpha_i}
  (1-\pi_i)^{|P_{i-1}(v)|\beta_i(1-\delta/8)}
  \le
  (1-\pi_i)^{\left(1-\frac{(1+2\delta)^{i-1}}{2}\right)p_{i-1}'\beta_i(1-\delta/8)}
  \le
  (1-\pi_i)^{\left(1-\frac{(1+2\delta)^{i-1}}{2}\right)p_i'}
  = \alpha_i.
\end{equation}
Therefore \(|N_{i,c}^S(u) \setminus S|\) is a sum of independent Bernoulli variables with mean \(\mu\le \alpha_i |\widetilde N_{i,c}^{S_a}(u) \setminus S|\). We claim
\[
\Prob_{\mathcal{V}_i^{1,b}}\left[|N_{i,c}^S(u) \setminus S| > \alpha_i \bigl|\widetilde N_{i,c}^{S_a}(u) \setminus S\bigr|\bigl(1+\frac{\delta}{3}\bigr)+\frac{\delta}{6}T \mid \mathcal{V}_i^{1,a}, \mathcal{E}(S_a)\right] \le d^{-34}.
\]
Indeed, if \(\mu\ge T\), then
\[
\alpha_i \bigl|\widetilde N_{i,c}^{S_a}(u) \setminus S\bigr| \bigl(1+\frac{\delta}{3}\bigr)+\frac{\delta}{6}T \ge \mu\bigl(1+\frac{\delta}{3}\bigr),
\]
so the Chernoff bound gives
\[
\Prob\!\left[|N_{i,c}^S(u) \setminus S| > \alpha_i |\widetilde N_{i,c}^{S_a}(u) \setminus S|\bigl(1+\frac{\delta}{3}\bigr)+\frac{\delta}{6}T\right]
\le \exp\!\left(-\frac{\delta^2\mu}{27}\right)\le \exp\!\left(-\frac{\delta^2T}{27}\right).
\]
If \(\mu<T\), then
\[
\alpha_i |\widetilde N_{i,c}^{S_a}(u) \setminus S|\bigl(1+\frac{\delta}{3}\bigr)+\frac{\delta}{6}T \ge \mu+\frac{\delta}{6}T.
\]
Applying Chernoff bound of the form \(\Prob[X>\mu+t]\le \exp\!\left(-\frac{t^2}{2(\mu+t/3)}\right)\) with \(X=|N_{i,c}^S(u) \setminus S|\) and \(t=\delta T/6\), and using \(\mu<T\) and \(\delta<1\), we get
\[
\Prob\left[|N_{i,c}^S(u) \setminus S| > \alpha_i |\widetilde N_{i,c}^{S_a}(u) \setminus S|\bigl(1+\frac{\delta}{3}\bigr)+\frac{\delta}{3}T\right]
\le
\exp\left(-\frac{(\delta T/6)^2}{2(\mu+\delta T/18)}\right)
\le
\exp\left(-\frac{\delta^2T}{74}\right).
\]
Since \(\delta=\log^{-4}\Delta\), \(T=\Delta^{\eps_1/3}\), and \(d\le \Delta^4\), we have \(\delta^2T\gg \log d\), so for \(\Delta\ge \Delta_\eps\),
$
\exp\left(-\frac{\delta^2T}{74}\right)\le d^{-34}
$, 
and the claim follows.

Moreover, the variables in $\bigcup_{v\in N_{i-1}(u) \cap S} \bigl( \mathrm{Sel}^2_i(v) \bigr)$ 
are also independent of $\mathcal{V}_i^{1,a}$ and the event $\mathcal{E}(S_a)$, and by \eqref{eq: ub by alpha_i}, for each $v\in N_{i-1}(u) \cap S$, the event $\mathrm{Sel}^{2}_i(v) \cap \widetilde{P}_i^{S_a}(v) = \emptyset$ happens with probability at most $\alpha_i$.
Through a similar argument, we can also show that
\[
\Prob_{\mathcal{V}_i^{2}}\left[|N_{i,c}^S(u) \cap S| > 
  \alpha_i \bigl|\widetilde N_{i,c}^{S_a}(u) \cap S\bigr|\bigl(1+\frac{\delta}{3}\bigr)+\frac{\delta}{6}T \mid \mathcal{V}_i^{1,a}, \mathcal{E}(S_a)\right] \le d^{-18}.
\]

Therefore, by a union bound, with probability at least \(1-2d^{-34}\) over the choice of the first sample \(\mathcal V_i^1 = V_i^{1,a} \cup V_i^{1,b}\), 
for every \(S\subseteq \Phi_j\) we have
\[
\Prob_{\mathcal V_i^2}\left[|N_{i,c}^S(u)| >  \alpha_i |\widetilde N_{i,c}^{S}(u)|\bigl(1+\frac{\delta}{3}\bigr)+\frac{\delta}{3}T \mid \mathcal V_i^1\right] < 3d^{-18},
\]
which completes the proof of the second claim and hence of the proposition.
\end{proof}

Before proving Proposition~\ref{prop:m-resilient}, define a signature map \(\Gamma\) that assigns to
each nonempty color set \(L\) a witness color \(\Gamma(L)\in L\). In the
proof below, whenever a vertex \(x\) comes with a nonempty set
\(L_x\) of colors still available to it, we use \(\Gamma(L_x)\) as the witness
color of \(x\).

\begin{proof}[Proof of Proposition~\ref{prop:m-resilient}]
Fix a vertex \(u\), a part \(\Phi_j\), and a subset \(S\subseteq \Phi_j\). 
Recall that we let \(N_i^\ell(u)\) be the residual neighborhood of \(u\) in the sample \(\mathcal V_i^\ell\), and let \(N_i^S(u)\) be the residual neighborhood of \(u\) in the mixed sample determined by \(S\). 
We consider the neighbors outside $\Phi_j$ that are newly added to the residual neighborhood by  resampling, namely
\(W_{i,j}^S(u):=N_i^S(u)\setminus \bigl(N_i^1(u)\cup \Phi_j\bigr)\).
Then every vertex of \(N_i^S(u)\) is either already in \(N_i^1(u)\), or lies in \(\Phi_j\cap N_{i-1}(u)\), or belongs to \(W_{i,j}^S(u)\). Therefore
\[
|N_i^S(u)|
\le
|N_i^1(u)|
\,+\,
|N_{i-1}(u)\cap \Phi_j|
\,+\,
|W_{i,j}^S(u)|.
\]

By Lemma~\ref{lem:appendix-degree-tail} from Appendix~\ref{sec:appendix-degree-tail},
\(\Prob_{\mathcal V_i^1}[M_i(u)]\le \Delta^{-200}\le d^{-40}\). 
Here \(M_i(u)\) is precisely the 
event \(\{|N_i^1(u)|>\alpha_i(1+\delta/5)d_{i-1}'\}\), so on its complement
the first term above satisfies
\(|N_i^1(u)|\le \alpha_i(1+\delta/5)d_{i-1}'\). Also, the partition bound gives
\(|N_{i-1}(u)\cap \Phi_j|\le 100d_{i-1}'/r\).

We next bound the size of \(W_{i,j}^S(u)\). 
For every vertex \(x\in N_{i-1}(u)\setminus \Phi_j\), define its first-sample certifying color
set \(C_x:=\mathrm{Sel}_i^1(x)\cap \widetilde P_i^1(x)\) and, whenever
\(C_x\neq\varnothing\), let \(\sigma(x)=\Gamma(C_x)\). If
\(x\in W_{i,j}^S(u)\), then \(x\notin N_i^1(u)\), so \(C_x\neq\varnothing\).
Moreover, \(x\) is uncolored in the mixed sample, so every color in \(C_x\) is
 selected by some vertex of
\(S\cap N_{i-1,\sigma(x)}(x)\).

For each color \(c'\), define
\begin{align*}
W_{i,j}^S(u,c')
:=
\bigl\{
x\in N_{i-1}(u)\setminus \Phi_j :\,&
C_x\neq\varnothing,\ \sigma(x)=c', \\
&c'\in \mathrm{Sel}_i^2(y)\text{ for some }
y\in S\cap N_{i-1,c'}(x)
\bigr\}.
\end{align*}
Then
\[
W_{i,j}^S(u)\subseteq \bigcup_{c'} W_{i,j}^S(u,c'),
\qquad\text{hence}\qquad
|W_{i,j}^S(u)|
\le
\sum_{c'} |W_{i,j}^S(u,c')|.
\]
For fixed \(c'\), the random variable \(|W_{i,j}^S(u,c')|\) depends only on the
\(c'\)-selection bits of the vertices in \(S\) from the second sample. Hence,
once the first sample \(\mathcal V_i^1\) has been fixed, the family
\(\{|W_{i,j}^S(u,c')|:c'\}\) is independent under
\(\Prob_{\mathcal V_i^2}[\cdot\mid \mathcal V_i^1]\).

To control the size of the set \(W_{i,j}^S(u,c')\),
define the following event in $\mathcal V_i^1$,
\[
O_i(u,c')
:=
\left\{
\bigl|\{x\in N_{i-1,c'}(u)\setminus \Phi_j:c'\in \mathrm{Sel}_i^1(x)\}\bigr|
>
(\log\Delta)^8
\right\}.
\]
Since \(|N_{i-1,c'}(u)|\le 2t_{i-1}\), the random variable inside
\(O_i(u,c')\) is binomial with mean at most \(2t_{i-1}\pi_i\le 1/K\). Thus a
standard binomial tail bound gives
\(\Prob_{\mathcal V_i^1}[O_i(u,c')]\le \exp(-\Omega((\log\Delta)^8))\).
There are at most \(\Delta\) relevant colors, so
\(\Prob_{\mathcal V_i^1}\left[\bigcup_{c'} O_i(u,c')\right]\le d^{-40}\) for \(\Delta\ge \Delta_\eps\). Moreover, if \(x\in W_{i,j}^S(u,c')\), then
\(\sigma(x)=c'\in C_x\subseteq \mathrm{Sel}_i^1(x)\) and \(x\in N_{i-1,c'}(u)\setminus
\Phi_j\). Therefore
\[
W_{i,j}^S(u,c')\subseteq
\{x\in N_{i-1,c'}(u)\setminus \Phi_j:\ c'\in \mathrm{Sel}_i^1(x)\},
\]
so on the complement of these events,
each set \(W_{i,j}^S(u,c')\) satisfies \(|W_{i,j}^S(u,c')|\le (\log\Delta)^8\).

Now fix a realization of the first sample \(\mathcal V_i^1\) on which neither
\(M_i(u)\) nor any event \(O_i(u,c')\) occurs. Under the conditional law
\(\Prob_{\mathcal V_i^2}[\cdot\mid \mathcal V_i^1]\), for a fixed vertex
\(x\in N_{i-1}(u)\setminus \Phi_j\), if \(\sigma(x)=c'\), then
\[
\Prob_{\mathcal V_i^2}
\left[x\in W_{i,j}^S(u,c')\mid \mathcal V_i^1\right]
\le
\pi_i\,|S\cap N_{i-1,c'}(x)|
\le
\frac{1}{2Kt_{i-1}+1}\cdot \frac{200t_{i-1}}{r}
\le
\frac{100}{Kr}.
\]
Moreover, since \(\sigma(x)\) is uniquely defined whenever \(C_x\neq\varnothing\),
the sets \(W_{i,j}^S(u,c')\) are pairwise disjoint in \(c'\). Hence
\[
\sum_{c'} |W_{i,j}^S(u,c')|
=
\sum_{\substack{x\in N_{i-1}(u)\setminus \Phi_j\\ C_x\neq\varnothing}}
\sum_{c'} \mathbf 1_{\{x\in W_{i,j}^S(u,c')\}},
\]
and for each \(x\) the inner sum has at most one nonzero term. Summing the
conditional probabilities over all \(x\in N_{i-1}(u)\setminus \Phi_j\) gives
\begin{align*}
\mu
:=
\E_{\mathcal V_i^2}\left[\sum_{c'} |W_{i,j}^S(u,c')| \,\middle|\, \mathcal V_i^1\right] 
\le \frac{100}{Kr}\,|N_{i-1}(u)|
\le \frac{100}{Kr}d_{i-1}'
\le \frac{25}{r}d_{i-1}'.
\end{align*}
Under the conditional law \(\Prob_{\mathcal V_i^2}[\cdot\mid \mathcal V_i^1]\),
the random variables \(\{|W_{i,j}^S(u,c')|\}_{c'}\) are independent, each lies
in \([0,(\log\Delta)^8]\) because \(\mathcal V_i^1\) avoids every event
\(O_i(u,c')\), and there are at most \(\Delta\) relevant colors \(c'\).
Applying Hoeffding's inequality to their sum therefore gives
\begin{align*}
\Prob_{\mathcal V_i^2}\left[\sum_{c'} |W_{i,j}^S(u,c')| > \frac{100}{r}d_{i-1}'+\frac{\Delta}{r(\log\Delta)^8} \middle|\, \mathcal V_i^1\right]
&\le
\exp\!\left(
-\frac{2\left(\frac{75}{r}d_{i-1}'+\frac{\Delta}{r(\log\Delta)^8}\right)^2}
{\Delta(\log\Delta)^{16}}
\right) \\
&\le d^{-20},
\end{align*}
because \(d_{i-1}'>d_{\mathrm{stop}}=\Delta/(\log\Delta)^6\) and
\(r=(\log\Delta)^{20}\).

Therefore, for this fixed \(\mathcal V_i^1\) and every
\(S\subseteq \Phi_j\),
\begin{align*}
\Prob_{\mathcal V_i^2}\left[(M_i^+(u))^S \,\middle|\, \mathcal V_i^1\right] 
&\le
\Prob_{\mathcal V_i^2}\left[
|N_i^S(u)|
>
\alpha_i\Bigl(1+\frac{\delta}{5}\Bigr)d_{i-1}'
+
\frac{200}{r}d_{i-1}'
+
\frac{\Delta}{r(\log\Delta)^8}
\middle|\, \mathcal V_i^1
\right]\\
&\le d^{-20}<d^{-10}.
\end{align*}
It follows that, for each fixed part \(\Phi_j\),
\[
\bigcup_{S\subseteq \Phi_j}
\left\{
\Prob_{\mathcal V_i^2}\left[(M_i^+(u))^S \,\middle|\, \mathcal V_i^1\right]
> d^{-10}
\right\}
\subseteq
M_i(u)\cup \bigcup_{c'} O_i(u,c'),
\]
whose probability is at most \(2d^{-40}\). Since
\(r=(\log\Delta)^{20}=d^{o(1)}\), a union bound over \(j\le r\) still leaves
total first-sample probability at most \(d^{-30}\), proving the resilience
condition for \(M_i^+(u)\).
\end{proof}

\subsection{Proof of Lemma~\ref{lem:triangle-free-round}}
\label{sec:triangle-free-round-proof}

\begin{lemma}
\label{lem:filtered-average}
With the parameter schedule fixed in Section~\ref{sec:triangle-free-parameters}, assume \(F_{i-1}\) holds and \(d_{i-1}'>d_{\mathrm{stop}}\). If \(Q_i^+(u)\) and \(R_i^+(u)\) fail to occur, then
\[
\widetilde n_i(u)\le \alpha_i\beta_i \overline n_{i-1}(u)+2\delta t_i.
\]
\end{lemma}

\begin{proof}
Fix \(u\in U_i\). 
Since  \(\widetilde P_i(u)\subseteq P_{i-1}(u)\),
\[
\widetilde n_i(u)
\;=\frac{1}{|\widetilde P_i(u)|}\sum_{c\in \widetilde P_i(u)} |N_{i,c}(u)|
\le \frac{1}{|\widetilde P_i(u)|}\sum_{c\in P_{i-1}(u)} |N_{i,c}(u)|.
\]
Since \(Q_i^+(u)\) and \(R_i^+(u)\) do not occur, we have
\[
|\widetilde P_i(u)|\ge |P_{i-1}(u)|\beta_i\bigl(1-\frac{\delta}{8}\bigr)\bigl(1-\frac{200}{r}\bigr)
\]
and
\[
\frac{\sum_{c\in P_{i-1}(u)} |N_{i,c}(u)|}{|P_{i-1}(u)|}
\le \alpha_i\Bigl(1+\frac{\delta}{3}\Bigr)
\left[\beta_i^2 \Bigl(\overline n_{i-1}(u)+\frac{\delta}{4}t_{i-1}\Bigr)
     + \frac{300}{r}t_{i-1}\right]
+ \frac{\delta}{3}T.
\]
Therefore
\[
\widetilde n_i(u)
\le \frac{\alpha_i\bigl(1+\frac{\delta}{3}\bigr)\left[\beta_i^2 \bigl(\overline n_{i-1}(u)+\frac{\delta}{4}t_{i-1}\bigr)+ \frac{300}{r}t_{i-1}\right]
+ \frac{\delta}{3}T}
{\beta_i(1-\delta/8)(1-200/r)}.
\]
Because \(d_{i-1}'>d_{\mathrm{stop}}\), we have \(i<\tau_{\mathrm{end}}\) by definition. Theorem~\ref{thm:executed-round-bounds} therefore applies to round \(i\), so \(i=O(k+\log\log\Delta)=O(\log\Delta)\), and hence
\[
t_{i-1}'=(1+2\delta)^{i-1}t_{i-1}\le 2t_{i-1}
\]
for \(\Delta\ge \Delta_\eps\). Also, Theorem~\ref{thm:executed-round-bounds} gives \(\beta_i\ge e^{-1/K}\), while the recurrence gives \(t_i\ge T\) and \(t_i\ge \alpha_i\beta_i t_{i-1}\). Note that
\[
\frac{1+\delta/3}{(1-\delta/8)(1-200/r)}\le 1+\frac{\delta}{2}.
\]
Also, the contribution of the \(\frac{\delta}{3}T\) term is
\[
\frac{\delta T}{3\beta_i(1-\delta/8)(1-200/r)}
\le
\frac{\delta}{3}e^{1/K}(1+o(1))\,T
\le
\frac{\delta}{2}T
\le
\frac{\delta}{2}t_i,
\]
where we use \(\beta_i\ge e^{-1/K}\), \((1-\delta/8)^{-1}(1-200/r)^{-1}=1+o(1)\), \(t_i\ge T\), and \(e^{1/K}\le e^{1/4}<3/2\) since \(K=4/\eps\ge 4\). 
Likewise,
the contribution of the \(\frac{300}{r}t_{i-1}\) term is at most \(\frac{\delta}{10}\,t_i\), because \(\beta_i\ge e^{-1/K}\) and \(1/r=o(\delta)\). Using also \(\overline n_{i-1}(u)\le 2t_{i-1}\), since every color in \(P_{i-1}(u)\) has residual \(c\)-degree at most \(2t_{i-1}\), we obtain
\begin{align*}
\widetilde n_i(u)
&\le \left(1+\frac{\delta}{2}\right)\alpha_i\beta_i
   \left(\overline n_{i-1}(u)+\frac{\delta}{4}t_{i-1}\right)
   + \frac{\delta}{10}t_i
   + \frac{\delta}{2}t_i \\
&\le \alpha_i\beta_i \overline n_{i-1}(u)
   + \frac{\delta}{2}\,\alpha_i\beta_i \overline n_{i-1}(u)
   + \frac{3\delta}{10}\,\alpha_i\beta_i t_{i-1}
   + \frac{\delta}{10}t_i
   + \frac{\delta}{2}t_i \\
&\le \alpha_i\beta_i \overline n_{i-1}(u)
   + \delta t_i
   + \frac{3\delta}{10}t_i
   + \frac{\delta}{10}t_i
   + \frac{\delta}{2}t_i \\
&\le \alpha_i\beta_i \overline n_{i-1}(u)+2\delta t_i,
\end{align*}
after enlarging \(\Delta\) if necessary. This is the required estimate.
\end{proof}

\begin{proof}[Proof of Lemma~\ref{lem:triangle-free-round}]
Assume \(F_{i-1}\) holds and \(d_{i-1}'>d_{\mathrm{stop}}\).
Also assume that the output of \(\mathrm{SolveRoundLLL}(i,r,\mathfrak L_i)\) avoids every bad event in \(\mathcal B_i\).
Fix \(u\in U_i\).

Because \(Q_i^+(u)\) does not occur and Theorem~\ref{thm:executed-round-bounds} gives \(p_{i-1}'\beta_i\Bigl(1-\frac{\delta}{8}\Bigr)\Bigl(1-\frac{200}{r}\Bigr)\ge p_i'\), we obtain
\[
\widetilde\lambda_i(u)
=\min\left\{1,\frac{|\widetilde P_i(u)|}{p_i'}\right\}
\ge \min\left\{1,\frac{|P_{i-1}(u)|}{p_{i-1}'}\right\}
= \lambda_{i-1}(u).
\]

Next, Lemma~\ref{lem:filtered-average} gives \(\widetilde n_i(u)\le \alpha_i\beta_i \overline n_{i-1}(u)+2\delta t_i\).
Since \(P_i(u)\) is obtained from \(\widetilde P_i(u)\) by first deleting every color whose residual degree exceeds \(2t_i\) and then truncating to at most \(p_i'\) colors, every color removed in this filtering-truncation step contributes at most the dummy value \(2t_i\). Therefore \(D_i(u)\le \widetilde D_i(u)\). Using also \(\widetilde\lambda_i(u)\ge \lambda_{i-1}(u)\), we may write
\begin{align*}
D_i(u)
&\le \widetilde D_i(u) \\
&= \widetilde\lambda_i(u)\widetilde n_i(u)+(1-\widetilde\lambda_i(u))2t_i \\
&\le \lambda_{i-1}(u)\widetilde n_i(u)+(1-\lambda_{i-1}(u))2t_i \\
&\le \lambda_{i-1}(u)\left(\alpha_i\beta_i \overline n_{i-1}(u)+2\delta t_i\right)
   +(1-\lambda_{i-1}(u))2t_i \\
&\le \frac{t_i}{t_{i-1}}\,D_{i-1}(u)+2\delta t_i \\
&\le (1+2\delta)^{i-1}t_i+2\delta t_i \\
&\le (1+2\delta)^i t_i=t_i',
\end{align*}
where the fourth line uses \(t_i\ge \alpha_i\beta_i t_{i-1}\), the next line uses \(F_{i-1}(u)\), and the last step uses \((1+2\delta)^{i-1}\ge 1\).

Finally, since \(M_i^+(u)\) does not occur, 
by Theorem~\ref{thm:executed-round-bounds}
\[
|N_i(u)|\le \alpha_i\Bigl(1+\frac{\delta}{3}\Bigr)d_{i-1}'
   + \frac{200}{r}d_{i-1}'
   + \frac{\Delta}{r(\log\Delta)^8}
\le d_i'.
\] 
Thus both parts of \(F_i(u)\) hold. Since \(u\) was arbitrary, \(F_i\) holds after the deterministic post-processing step.
\end{proof}

\section{Girth-Five Graphs}

This section proves Theorem~\ref{thm:girth 5} stated in the introduction. 
The proof follows the same resilient-\LLL framework as Section~3, 
but random coloring on girth-$5$ graphs gives enough local independence to maintain 
direct per-color \(c\)-degree upper bounds instead of the average \(c\)-degree upper bounds used in the triangle-free case. 
We therefore define the more suitable round-\(i\) instance, invariant, and simplified iterative algorithm together, then prove the theorem.

\subsection{The Round-\texorpdfstring{$i$}{i} Instance, Invariant, and Algorithm}
\label{sec:g5-round-instance}

Fix a round \(i\ge 1\). Let \(G_{i-1}\) be the graph induced by the vertices still uncolored at the start of round \(i\). Let \(P_{i-1}(u)\) be the palette of an active vertex \(u\), let \(N_{i-1}(u)\) be its neighborhood in \(G_{i-1}\), and let \(N_{i-1,c}(u)=\{v\in N_{i-1}(u): c\in P_{i-1}(v)\}\). As in Section~3, the round variables are \(X_i(u)=(\mathrm{Sel}_i(u),K_i(u))\), sampled independently by Algorithm~\ref{alg:select}, and the round-\(i\) variable family is \(\mathcal V_i=\{X_i(u):u\in V(G_{i-1})\}\). Let \(\theta_i=(\pi_i,\beta_i,\alpha_i,p_i,p_i',t_i,t_i',d_i,d_i')\) denote the deterministic parameter bundle for round \(i\), and let \(\mathfrak L_i=(\mathcal B_i,\mathcal V_i,\theta_i)\) be the associated local \LLL instance.

For a vertex \(u\) still active after round \(i\), let \(N_i(u)\) be its neighborhood in \(G_i\), and let \(F_i(u)\) be the event that
\begin{enumerate}[leftmargin=*,itemsep=0.25em]
  \item \(|P_i(u)|\ge p_i'\),
  \item \(|N_{i,c}(u)|\le t_i'\) for every color \(c\in P_i(u)\),
  \item \(|N_i(u)|\le d_i'\).
\end{enumerate}
Let \(F_i\) be the event that \(F_i(u)\) holds for every vertex still present in \(G_i\). 

We now state the streamlined iterative algorithm~\ref{alg:g5-framework}. 
In each round, the vertices sample \((\mathrm{Sel}_i(u),K_i(u))\), solve the associated local \LLL instance, delete direct color conflicts, color every vertex that still keeps one of its selected colors, and truncate each surviving palette to the target size \(p_i'\). 

\begin{algorithm}[t]
\caption{IterativeColoringFramework$(G,\{\theta_i\}_{i\ge 0},r)$}
\label{alg:g5-framework}
\begin{algorithmic}[1]
\State Set $G_0=G$ and initialize the palettes $P_0(u)$ to be of $p_0$ colors for every vertex $u$
\State $d_{\mathrm{stop}}\gets \Delta/(\log\Delta)^6$
\State $i\gets 1$
\While{$d_{i-1}'>d_{\mathrm{stop}}$}
  \State construct the round-$i$ \LLL instance $\mathfrak L_i$
  \State $X_i\gets \mathrm{SolveRoundLLL}(i,r,\mathfrak L_i)$
  \ForAll{vertices $u\in V(G_{i-1})$}
    \State $\widehat P_i(u)\gets K_i(u)\setminus \mathrm{Sel}_i(N_{i-1}(u))$
    \If{$\mathrm{Sel}_i(u)\cap \widehat P_i(u)\neq\emptyset$}
      \State color $u$ with an arbitrary color in $\mathrm{Sel}_i(u)\cap \widehat P_i(u)$
    \EndIf
  \EndFor
  \State let $U_i$ be the vertices of $G_{i-1}$ still uncolored after the round-$i$ coloring decisions
  \ForAll{vertices $u\in U_i$}
    \State set $P_i(u)$ to an arbitrary subset of $\widehat P_i(u)$ of size exactly $p_i'$
  \EndFor
  \State let $G_i$ be the graph induced by $U_i$
  \State $i\gets i+1$
\EndWhile
\State color the remaining graph $G_{i-1}$ using Theorem~\ref{thm:clp20} with $d_{\mathrm{stop}}+1$ fresh colors
\end{algorithmic}
\end{algorithm}

The family \(\mathcal B_i\) of bad events is the following. 
\begin{align*}
Q_i^+(u)
& := \left\{|\widehat P_i(u)|<p_{i-1}'\beta_i\Bigl(1-\frac{\delta}{8}\Bigr)\Bigl(1-\frac{200}{r}\Bigr)\right\},\\
H_i^+(u,c)
& := \left\{|N_{i,c}(u)|>t_{i-1}'\left((1+\delta)\alpha_i\beta_i'+\frac{300}{r}\right)+\delta T\right\},\\
M_i^+(u)
& := \left\{|N_i(u)|
  > \alpha_i\Bigl(1+\frac{\delta}{3}\Bigr)d_{i-1}'
    + \frac{200}{r}d_{i-1}'
    + \frac{\Delta}{r(\log\Delta)^8}\right\}.
\end{align*}

Fix a constant \(\eps>0\), choose a parameter \(k\) with \(1\le k\le (1-\eps)\ln\Delta\), and define
\[
\eps_2=\frac{\eps}{5},
\qquad
\eps_1=\eps_2\left(1-\frac{\eps}{3}\right),
\qquad
K=\frac{4}{\eps}.
\]
Set
\[
\delta=\frac{1}{\log^4\Delta},
\qquad
r=(\log\Delta)^{20},
\]
and let \(T=\Delta^{\eps_1}/3\). 
The target palette size and its relaxed lower bound are
\[
\qquad
p_0=\frac{\Delta}{k}-\frac{\Delta}{(\log\Delta)^6}-1,
\qquad
p_i=\beta_i p_{i-1},
\qquad
p_i'=(1-\delta/7)^i p_i.
\]
The sampling probability is
\[
\pi_i=\frac{1}{Kt_{i-1}'+1},
\]
the color-retention factor is
\[
\beta_i=(1-\pi_i)^{t_{i-1}'},
\qquad
\beta_i'=\frac{\beta_i}{1-\pi_i},
\]
and the residual-degree attenuation factor is
\[
\alpha_i=(1-\pi_i)^{p_i'}.
\]
We then define
\[
t_0=\Delta,\qquad t_i=\max\{\alpha_i\beta_i' t_{i-1},T,p_i\},\qquad t_i'=(1+2\delta)^i t_i,
\]
and
\[
d_0=d_0'=\Delta,\qquad d_i=\max\!\left\{\alpha_i d_{i-1},\frac{\Delta}{(1+\delta)^i(\log\Delta)^6}\right\},\qquad d_i'=(1+\delta)^i d_i.
\]

\subsection{Main Lemmas and Proof of the Girth-\texorpdfstring{$5$}{5} Theorem}

We now state the three ingredients needed for Theorem~\ref{thm:girth 5}. 
They are the girth-$5$ counterparts of Lemmas~\ref{lem:triangle-free-instance}, \ref{lem:triangle-free-round}, and~\ref{lem:parameter-evolution} from Section~\ref{sec:triangle-free-main-lemmas}. 
As in Section~\ref{sec:triangle-free}, the first lemma gives a fast solver for each round-\(i\) local \LLL instance, the second shows that avoiding the bad events propagates the round invariant \(F_i\), and the third shows that the degree-bound sequence \(d_i'\) reaches the stopping threshold quickly. 

\begin{lemma}
\label{lem:g5-instance}
With the parameter schedule fixed in Section~\ref{sec:g5-round-instance}, for every fixed \(C>0\), the following holds whenever \(\Delta\le(\log n)^C\). Suppose \(F_{i-1}\) holds and \(d_{i-1}'>d_{\mathrm{stop}}\). Then, with high probability, Algorithm~\ref{alg:solve-lll} solves the round-\(i\) instance \(\mathfrak L_i\) in \((\log\log n)^{O(1)}\) rounds.
\end{lemma}

\begin{lemma}
\label{lem:g5-round}
With the parameter schedule fixed in Section~\ref{sec:g5-round-instance}, assume \(F_{i-1}\) holds and \(d_{i-1}'>d_{\mathrm{stop}}\). For any realization \(X_i\) of the round-\(i\) variables in which no bad event of \(\mathcal B_i\) occurs, the deterministic post-processing step in Algorithm~\ref{alg:g5-framework} yields \(F_i\).
\end{lemma}

\begin{lemma}
\label{lem:g5-parameter}
With the parameter schedule fixed in Section~\ref{sec:g5-round-instance}, there is \(\ell=O(k+\log\log\Delta)\) such that
\[
p_j>T\quad\text{for every } 1\le j\le \ell
\qquad\text{and}\qquad
d_\ell'\le d_{\mathrm{stop}}=\frac{\Delta}{(\log\Delta)^6}.
\]
\end{lemma}

Theorem~\ref{thm:girth 5} is deduced from these lemmas in exactly the same way that Theorem~\ref{thm:triangle-free-main} is deduced from Lemmas~\ref{lem:triangle-free-instance}, 
\ref{lem:triangle-free-round}, and~\ref{lem:parameter-evolution} in the proof from Section~\ref{sec:triangle-free-main-lemmas}. 
First, Lemma~\ref{lem:g5-parameter} gives a round \(\ell=O(k+\log\log\Delta)\) with \(d_\ell'\le d_{\mathrm{stop}}\). 
Next, Lemmas~\ref{lem:g5-instance} and~\ref{lem:g5-round} show that the rounds \(1,\dots,\ell\) can be executed while maintaining \(F_i\). 
Consequently the residual graph \(G_\ell\) has maximum degree at most \(d_{\mathrm{stop}}\), so Theorem~\ref{thm:clp20} completes the coloring.

We therefore omit the proof of Theorem~\ref{thm:girth 5} here and instead prove the three lemmas.

For Lemma~\ref{lem:g5-instance}, the key ingredient is that the three girth-$5$ bad-event families are resilient with respect to a suitable partition. Accordingly, fix an \(r\)-partition \(\Phi\) with degree parameters \(\min\{d_{i-1}',\Delta\}\) and \(t_{i-1}'\) in the sense of Definition~\ref{def:r-partition-degree-params} from Section~\ref{sec:color-balanced-partition}. We first state the corresponding resilience propositions, whose proofs are deferred to the next subsection; once they are proved, the proof of Lemma~\ref{lem:g5-instance} is the same as the proof of Lemma~\ref{lem:triangle-free-instance} from Section~\ref{sec:triangle-free-lemma-proofs}, so we omit it.

\begin{proposition}
\label{prop:g5-q-resilient}
Assume \(F_{i-1}\) holds and \(d_{i-1}'>d_{\mathrm{stop}}\). Then for every active vertex \(u\), the bad event \(Q_i^+(u)\) is resilient with respect to \(\Phi\).
\end{proposition}

\begin{proposition}
\label{prop:g5-h-resilient}
Assume \(F_{i-1}\) holds and \(d_{i-1}'>d_{\mathrm{stop}}\). Then for every active vertex \(u\) and every color \(c\in P_{i-1}(u)\), the bad event \(H_i^+(u,c)\) is resilient with respect to \(\Phi\).
\end{proposition}

\begin{proposition}
\label{prop:g5-m-resilient}
Assume \(F_{i-1}\) holds and \(d_{i-1}'>d_{\mathrm{stop}}\). Then for every active vertex \(u\), the bad event \(M_i^+(u)\) is resilient with respect to \(\Phi\).
\end{proposition}

\begin{proof}[Proof of Lemma~\ref{lem:g5-round}]
Assume \(F_{i-1}\) holds, \(d_{i-1}'>d_{\mathrm{stop}}\), and no bad event of \(\mathcal B_i\) occurs. We verify the three clauses of \(F_i\).

First, \(Q_i^+(u)\) does not occur, so
\[
|\widehat P_i(u)|
\ge
p_{i-1}'\beta_i\Bigl(1-\frac{\delta}{8}\Bigr)\Bigl(1-\frac{200}{r}\Bigr).
\]
By the definition \(p_i'=(1-\delta/7)^ip_i\) and \(p_i=\beta_i p_{i-1}\), the right-hand side is at least \(p_i'\) for \(\Delta\ge \Delta_\eps\), because \(200/r=o(\delta)\). Therefore the truncation step in Algorithm~\ref{alg:g5-framework} is well-defined, and after truncation \(|P_i(u)|=p_i'\).

Second, let \(c\in P_i(u)\). 
Because \(H_i^+(u,c)\) does not occur, we obtain
\begin{equation}
\label{eq:g5-round-c-degree-threshold}
|N_{i,c}(u)|
\le
t_{i-1}'\left((1+\delta)\alpha_i\beta_i'+\frac{300}{r}\right)+\delta T.
\end{equation}
To compare the right-hand side of \eqref{eq:g5-round-c-degree-threshold} with \(t_i'\), note that \(t_{i-1}'=(1+2\delta)^{i-1}t_{i-1}\) and \(t_i'=(1+2\delta)^it_i\). Also
\(\pi_i=(Kt_{i-1}'+1)^{-1}\le (KT)^{-1}=o(\delta)\). Hence, for \(\Delta\ge \Delta_\eps\), we have \(\pi_i\le 1/2\),
\[
\beta_i'=\frac{\beta_i}{1-\pi_i}\le \beta_i e^{2\pi_i}=\beta_i(1+o(\delta))\]
Moreover,
\[
\beta_i=(1-\pi_i)^{t_{i-1}'}
\ge
e^{-2\pi_i t_{i-1}'}
\ge
e^{-2/K}
\qquad\text{and}\qquad
\alpha_i=(1-\pi_i)^{p_i'}
\ge
e^{-2\pi_i p_i'}
\ge
e^{-2/K},
\]
because \(t_{i-1}'\pi_i\le 1/K\) and \(p_i'\le t_{i-1}'\). Hence \(\alpha_i\beta_i'\ge \alpha_i\beta_i\ge e^{-4/K}\).
Therefore,
\begin{align*}
|N_{i,c}(u)|
&\le
t_{i-1}'\left((1+\delta)\alpha_i\beta_i'+\frac{300}{r}\right)+\delta T,
&& \text{by \eqref{eq:g5-round-c-degree-threshold},}\\
&\le
(1+2\delta)^{i-1}\left((1+\delta)\alpha_i\beta_i' t_{i-1}+\frac{300t_{i-1}}{r}\right)+\delta T,
&& \text{using } t_{i-1}'=(1+2\delta)^{i-1}t_{i-1},\\
&\le
(1+2\delta)^{i-1}\bigl((1+2\delta)\alpha_i\beta_i' t_{i-1}+\delta T\bigr),
&& \text{since } \alpha_i\beta_i'\ge e^{-4/K} \\
& &&\text{and } 300r^{-1}=o(\delta),\\
&\le
(1+2\delta)^it_i,
&& \text{because } t_i\ge \alpha_i\beta_i' t_{i-1},\ T,\\
&=
t_i',
&& \text{by definition of } t_i'.
\end{align*}

Finally, we consider the upper bound of \(|N_i(u)|\). Since \(M_i^+(u)\) does not occur, we have
\begin{align*}
|N_i(u)|
&\le
\alpha_i\Bigl(1+\frac{\delta}{3}\Bigr)d_{i-1}'
+\frac{200}{r}d_{i-1}'
+\frac{\Delta}{r(\log\Delta)^8},\\
&\le
\alpha_i\Bigl(1+\frac{\delta}{3}\Bigr)d_{i-1}'
+o(\delta\alpha_i d_{i-1}'),\\
&\le
(1+\delta)\alpha_i d_{i-1}',\\
&\le
d_i'.
\end{align*}

All three defining properties of \(F_i(u)\) hold for every active vertex \(u\), so \(F_i\) holds.
\end{proof}

\begin{proof}[Proof of Lemma~\ref{lem:g5-parameter}]
Let \(s_i=t_i/p_i\), and let \(\tau=\min\{i\ge 0:p_i\ge t_i\}\).
Since \(t_i\ge p_i\) for every \(i\) by definition of \(t_i\), this is equivalently
\(\tau=\min\{i\ge 0:t_i=p_i\}\).
Exactly as in the proof of Theorem~\ref{thm:executed-round-bounds} from Section~\ref{sec:triangle-free-numerical-proof}, 
the same recurrence for \(s_i\) yields
\(\tau=O(k+\log\log\Delta)\),
and, for every fixed constant \(A\) and every \(j\in[\tau,\tau+A\log\log\Delta]\),
\begin{equation}
\label{eq:g5-post-tau-p-lower}
p_j\ge \Delta^{\eps_2-o(1)}>T.
\end{equation}
We do not repeat the same analysis here. 

We now analyze the rounds after \(\tau\). By definition of \(\tau\), we have \(t_\tau=p_\tau\). We claim that
\(t_j=p_j\) for every \(j\in[\tau,\tau+A\log\log\Delta]\).
The base case \(j=\tau\) is immediate. If \(\tau\le j<\tau+A\log\log\Delta\) and \(t_j=p_j\), then \eqref{eq:g5-post-tau-p-lower} gives \(p_{j+1}>T\), and therefore
\[
t_{j+1}
=\max\{\alpha_{j+1}\beta_{j+1}t_j,T,p_{j+1}\}
=\max\{\alpha_{j+1}p_{j+1},T,p_{j+1}\}
=p_{j+1},
\]
since \(\alpha_{j+1}\le 1\). This proves the claim.

Now fix \(i\) with \(\tau+1\le i\le\tau+A\log\log\Delta\). Since \(i=O(k+\log\log\Delta)=O(\log\Delta)\), we have \(p_i'=(1-o(1))p_i\) and \(t_{i-1}'=(1+o(1))t_{i-1}\). Using \(t_{i-1}=p_{i-1}\), and writing \(\pi_i=(Kt_{i-1}'+1)^{-1}\), we get \(\pi_i=o(1)\) and therefore
\[
t_{i-1}'\log(1-\pi_i)
=-t_{i-1}'\pi_i+O(t_{i-1}'\pi_i^2)
=-\frac{t_{i-1}'}{Kt_{i-1}'+1}+o(1)
=-\frac{1}{K}+o(1).
\]
Therefore
\[
\frac{p_i}{t_{i-1}}
=\beta_i
=(1-\pi_i)^{t_{i-1}'}
=\exp\!\left(-\frac{1}{K}+o(1)\right).
\]
Consequently,
\[
\alpha_i
=(1-\pi_i)^{p_i'}
\le \exp\!\left(-\frac{(1-o(1))p_i}{Kt_{i-1}}\right)
\le \exp\!\left(-\frac{e^{-1/K}-o(1)}{K}\right).
\]
Therefore, there is a constant \(c:=e^{-1/K}/(2K)>0\) such that \(\alpha_i\le e^{-c}\) throughout this interval. Since also \((1+\delta)e^{-c}\le e^{-c/2}\) for \(\Delta\ge \Delta_\eps\), we get
\[
d_i'
=\max\!\left\{(1+\delta)\alpha_i d_{i-1}',d_{\mathrm{stop}}\right\}
\le \max\!\left\{e^{-c/2}d_{i-1}',d_{\mathrm{stop}}\right\}.
\]
for every \(\tau+1\le i\le\tau+A\log\log\Delta\). Consequently, until the sequence first reaches \(d_{\mathrm{stop}}\), it shrinks geometrically by the factor \(e^{-c/2}\). Since \(\alpha_j\le 1\) for every \(j\), the recurrence for \(d_j'\) gives \(d_j'\le (1+\delta)^j\Delta\) for all \(j\ge 0\). Because \(\tau=O(k+\log\log\Delta)=O(\log\Delta)\) and \(\delta=\log^{-4}\Delta\), we have \(d_\tau'\le (1+\delta)^\tau\Delta=(1+o(1))\Delta\). Therefore
\[
d_{\tau+s}'\le (1+o(1))e^{-cs/2}\Delta
\]
for every \(s\ge 0\) for which the rounds \(\tau+1,\dots,\tau+s\) are still executed. Choosing \(A>14/c\), we obtain
\[
d_{\tau+A\log\log\Delta}'
\le \frac{\Delta}{(\log\Delta)^7}
< \frac{\Delta}{(\log\Delta)^6}
=d_{\mathrm{stop}}.
\]
Let \(\ell\) be the first round with \(d_\ell'\le d_{\mathrm{stop}}\). Then \(\ell\le \tau+A\log\log\Delta=O(k+\log\log\Delta)\), and \eqref{eq:g5-post-tau-p-lower} implies \(p_j>T\) for every \(j\in[\tau,\ell]\). Since \(p_j=\beta_jp_{j-1}\) and \(0<\beta_j\le 1\) for every \(j\), the sequence \((p_j)_{j\ge 0}\) is nonincreasing. Therefore \(p_j\ge p_\ell>T\) for every \(1\le j\le \ell\). This proves the lemma.
\end{proof}

\subsection{Proofs of the Resilience Propositions}
\label{sec:g5-resilience}

Throughout this subsection, fix a round \(i\), assume \(F_{i-1}\) holds and \(d_{i-1}'>d_{\mathrm{stop}}\), and let \(\Phi=(\Phi_1,\ldots,\Phi_r)\) be an \(r\)-partition with parameters \(\min\{d_{i-1}',\Delta\}\) and \(t_{i-1}'\) in the sense of Definition~\ref{def:r-partition-degree-params} from Section~\ref{sec:color-balanced-partition}. Thus, for every active vertex \(u\), every color \(c\in P_{i-1}(u)\), and every part \(\Phi_j\),
\[
|N_{i-1}(u)\cap \Phi_j|
\le \frac{100\min\{d_{i-1}',\Delta\}}{r}
\le \frac{100d_{i-1}'}{r}
\qquad\text{and}\qquad
|N_{i-1,c}(u)\cap \Phi_j|
\le \frac{100t_{i-1}'}{r}.
\]
Let \(\mathcal V_i^1\) and \(\mathcal V_i^2\) be independent copies of the round-\(i\) variables, and for \(S\subseteq \Phi_j\) let \(Z_i^S\) denote the mixed-sample version of any round-\(i\) random object \(Z_i\), as in Section~\ref{sec:resilient-lll}. We write \(d\) for the dependency degree of the round-\(i\) event graph, so \(d\le \Delta^4\).

The proofs of Propositions~\ref{prop:g5-q-resilient} and~\ref{prop:g5-m-resilient}
follow the proofs of Propositions~\ref{prop:q-resilient}
and~\ref{prop:m-resilient} from Section~\ref{sec:triangle-free-lemma-proofs},
with the necessary changes in the constants involved. We therefore focus only
on Proposition~\ref{prop:g5-h-resilient}, which is the only genuinely new point
in the girth-$5$ case.
Fix once and for all a signature map \(\sigma\) that chooses one color from each
nonempty set of colors.
\begin{proof}[Proof of Proposition~\ref{prop:g5-h-resilient}]
Fix \(u\), \(c\in P_{i-1}(u)\), a part \(\Phi_j\), and \(S\subseteq \Phi_j\).
For this proof only, let
\[
\widehat N_{i,c}(u):=
\{x\in N_{i-1,c}(u)\cap U_i:c\in \widehat P_i(x)\}.
\]
Since every vertex of \(N_{i,c}(u)\) lies in \(N_{i-1,c}(u)\cap U_i\) and
satisfies \(P_i(x)\subseteq \widehat P_i(x)\), we have
\(N_{i,c}(u)\subseteq \widehat N_{i,c}(u)\). Hence it is enough to prove
resilience for
\[
  \widehat H_i^+(u,c):=
  \left\{
  |\widehat N_{i,c}(u)|
  >
  t_{i-1}'\left((1+\delta)\alpha_i\beta_i'+\frac{300}{r}\right)+\delta T
  \right\}.
\]
Let
\[
\widehat H_i(u,c):=
\left\{
|\widehat N_{i,c}(u)|
>
t_{i-1}'(1+\delta)\alpha_i\beta_i'+\delta T
\right\}.
\]
By Lemma~\ref{lem:appendix-g5-degree-tail} from
Appendix~\ref{sec:appendix-degree-tail}, we have
\(\Prob[\widehat H_i(u,c)]\le \Delta^{-200}\) for \(\Delta\ge \Delta_\eps\). Define
\[
\widetilde H_i(u,c):=
\widehat H_i(u,c)\cup
\left\{
\Prob_{\mathcal V_i^2}\left[(\widehat H_i(u,c))^{\{u\}}\right]\ge d^{-15}
\right\}.
\]
Since the mixed sample with resampling set \(\{u\}\) has the same distribution
as the original round-\(i\) sample, we have
\[
\E_{\mathcal V_i^1}
\left[
\Prob_{\mathcal V_i^2}\left[(\widehat H_i(u,c))^{\{u\}} \middle| \mathcal{V}_i^1\right]
\right]
=
\Prob[\widehat H_i(u,c)]
\le \Delta^{-200}.
\]
Hence Markov's inequality gives
\[
\Prob_{\mathcal V_i^1}\left[
\Prob_{\mathcal V_i^2}
\left[(\widehat H_i(u,c))^{\{u\}} \middle| \mathcal{V}_i^1\right]\ge d^{-15}
\right]
\le d^{15}\Delta^{-200}
\le d^{-35},
\]
because \(d\le \Delta^4\). Therefore
\[
\Prob_{\mathcal V_i^1}[\widetilde H_i(u,c)]
\le \Prob[\widehat H_i(u,c)] + d^{-35}
\le 2d^{-35}.
\]

For any round-\(i\) random object \(Z_i\), let \(Z_i^\star\) denote the sample
obtained by first resampling \(u\) if \(u\in S\), and otherwise leaving the
first sample unchanged. Thus \(Z_i^\star\) is either the first sample itself or
the mixed sample in which only \(u\) is resampled. 
If \(u\in S\), fix an arbitrary realization \(\xi_u\) of \(X_i^2(u)\); if
\(u\notin S\), no additional conditioning is needed. We compare
\(\widehat N_{i,c}^S(u)\) with \(\widehat N_{i,c}^\star(u)\). Every vertex in
\(\widehat N_{i,c}^S(u)\setminus \widehat N_{i,c}^\star(u)\) is produced by one
of the following three effects.

The first effect is the direct contribution of the resampled block itself:
\[
\Delta_1^S:=\Phi_j\cap \widehat N_{i,c}^S(u).
\]
By the definition of an $r$-partition,
\[
|\Delta_1^S|
\le |\Phi_j\cap N_{i-1,c}(u)|
\le \frac{100t_{i-1}'}{r}.
\]
The second effect is that a vertex \(x\in N_{i-1,c}(u)\setminus \Phi_j\) may
satisfy \(c\in \widehat P_i^S(x)\setminus \widehat P_i^\star(x)\). Define
\[
\Delta_2^S:=
\left\{
x\in \widehat N_{i,c}^S(u)\setminus \widehat N_{i,c}^\star(u):
x\notin \Phi_j,\ c\notin \widehat P_i^\star(x)
\right\},
\]
and
\[
A_{i,j}(u,c):=
\left\{
x\in N_{i-1,c}(u)\setminus \Phi_j:
c\in \mathrm{Sel}_i^1(y)\text{ for some }y\in
(\Phi_j\setminus\{u\})\cap N_{i-1,c}(x)
\right\}.
\]
Then \(\Delta_2^S\subseteq A_{i,j}(u,c)\), because the \(\star\)-sample and the
\(S\)-sample agree outside \(S\setminus\{u\}\subseteq \Phi_j\setminus\{u\}\), and therefore
the only way for a vertex \(x\notin \Phi_j\) to satisfy
\(c\in \widehat P_i^S(x)\setminus \widehat P_i^\star(x)\) is that some vertex
of \((\Phi_j\setminus\{u\})\cap N_{i-1,c}(x)\) selected \(c\) in the
\(\star\)-sample.

For each fixed \(x\in N_{i-1,c}(u)\setminus \Phi_j\),
\[
\Prob_{\mathcal V_i^1}[x\in A_{i,j}(u,c)]
\le
\pi_i\cdot |(\Phi_j\setminus\{u\})\cap N_{i-1,c}(x)|
\le
\frac{100}{Kr}.
\]
Moreover, if \(x\neq x'\) are distinct
vertices in \(N_{i-1,c}(u)\setminus \Phi_j\), then
\[
\bigl(N_{i-1,c}(x)\cap(\Phi_j\setminus\{u\})\bigr)
\cap
\bigl(N_{i-1,c}(x')\cap(\Phi_j\setminus\{u\})\bigr)
=\emptyset,
\]
for otherwise some vertex \(z\in \Phi_j\setminus\{u\}\) adjacent to both \(x\)
and \(x'\) would create a cycle of length 4 on the vertices \(\{u,x,z,x'\}\). Hence the indicators of
the events \(x\in A_{i,j}(u,c)\) are independent. 
Let
\[
J_{i,j}(u,c):=
\left\{
  \left|A_{i,j}(u,c)\right|>\frac{100t_{i-1}'}{r}
\right\}.
\]
Since \(|N_{i-1,c}(u)|\le t_{i-1}'\), the random variable \(|A_{i,j}(u,c)|\) is
a sum of independent Bernoulli variables with mean
\(\mu_A:=\E_{\mathcal V_i^1}[|A_{i,j}(u,c)|]\le 100t_{i-1}'/(Kr)\). Hence
\(100t_{i-1}'/r\ge K\mu_A\). 
Applying the Chernoff bound in the
form \(\Prob[X\ge a\mu]\le (e/a)^{a\mu}\) for \(a\ge 1\), with
\(X=|A_{i,j}(u,c)|\), \(\mu=\mu_A\), and \(a=K\), gives
\[
\Prob_{\mathcal V_i^1}[J_{i,j}(u,c)]
\le \left(\frac{e}{K}\right)^{100t_{i-1}'/r}
\le d^{-35}
\]
for \(\Delta\ge \Delta_\eps\), because \(t_{i-1}'\ge T=\Delta^{\eps_1}/3\),
\(r=(\log\Delta)^{20}\), and \(d\le \Delta^4\).

The third effect is that a vertex \(x\in N_{i-1,c}(u)\setminus \Phi_j\) may
fail to belong to \(\widehat N_{i,c}^\star(u)\) but belong to
\(\widehat N_{i,c}^S(u)\) because the resampling of \(S\setminus\{u\}\) causes
\(x\) to become uncolored. Define
\[
\Delta_3^S:=
\left\{
x\in \widehat N_{i,c}^S(u)\setminus \widehat N_{i,c}^\star(u):
x\notin \Phi_j\cup A_{i,j}(u,c)
\right\}.
\]
For each \(x\in N_{i-1,c}(u)\setminus \Phi_j\), let
\[
L_x:=
\mathrm{Sel}_i^\star(x)\cap
\Bigl(
K_i^\star(x)\setminus
\mathrm{Sel}_i^\star\bigl((N_{i-1}(x)\setminus \Phi_j)\cup\{u\}\bigr)
\Bigr).
\]
Thus \(L_x\) is the set of colors selected by \(x\) that remain available after
deleting the conflicts created by the neighbors of \(x\) in
\((N_{i-1}(x)\setminus \Phi_j)\cup\{u\}\). Equivalently, the only vertices that
can still delete a color from \(L_x\) are the vertices of
\((\Phi_j\setminus\{u\})\cap N_{i-1}(x)\). Write \(\sigma(x):=\sigma(L_x)\)
whenever \(L_x\neq\emptyset\). If
\(L_x=\emptyset\), set \(Y_x^S:=0\); otherwise, let \(Y_x^S\) be the indicator
of the event that some vertex of \((S\setminus\{u\})\cap N_{i-1,\sigma(x)}(x)\) selects
\(\sigma(x)\) in the second sample; that is,
\[
Y_x^S=
\mathbf 1\left[\sigma(x)\in \mathrm{Sel}_i^S(y)\text{ for some }y\in (S\setminus\{u\})\cap N_{i-1,\sigma(x)}(x)\right].
\]

If \(x\in \Delta_3^S\), then \(x\notin A_{i,j}(u,c)\), so \(c\in \widehat
P_i^\star(x)\). Since also \(x\notin \widehat N_{i,c}^\star(u)\), it follows
that \(x\) is colored in the \(\star\)-sample. Hence
\(\mathrm{Sel}_i^\star(x)\cap \widehat P_i^\star(x)\neq\emptyset\), and every
color in this set belongs to \(L_x\); in particular, \(L_x\neq\emptyset\).
On the other hand, \(x\in \widehat N_{i,c}^S(u)\), so \(x\) is uncolored in the
\(S\)-sample. Therefore the color \(\sigma(x)\in L_x\) is deleted in the
\(S\)-sample, which means that
\(\sigma(x)\in \mathrm{Sel}_i^S(y)\) for some
\(y\in (S\setminus\{u\})\cap N_{i-1,\sigma(x)}(x)\). Thus \(Y_x^S=1\), and
hence
\[
\Delta_3^S\subseteq \{x:Y_x^S=1\}.
\]

For each \(x\) with \(L_x\neq\emptyset\), we have
\[
\Prob_{\mathcal V_i^2}\left[
Y_x^S=1
\mid
\mathcal V_i^1,\ X_i^2(u)=\xi_u
\right]
\le
\pi_i\cdot |(S\setminus\{u\})\cap N_{i-1,\sigma(x)}(x)|
\le
\pi_i\cdot \frac{100t_{i-1}'}{r}
\le
\frac{100}{Kr}.
\]
Also, since $G$ is girth-5, the indicators \(Y_x^S\) are independent. Let
\[
Y^S:=\sum_{x\in N_{i-1,c}(u)\setminus \Phi_j} Y_x^S.
\]
Then
\(\E_{\mathcal V_i^2}[Y^S\mid \mathcal V_i^1,\ X_i^2(u)=\xi_u]
\le 100t_{i-1}'/(Kr)\). Writing
\(\mu_Y:=\E_{\mathcal V_i^2}[Y^S\mid \mathcal V_i^1,\ X_i^2(u)=\xi_u]\), we
therefore have \(100t_{i-1}'/r\ge K\mu_Y\). Applying the same
Chernoff bound with \(X=Y^S\), \(\mu=\mu_Y\), and \(a=K\), we obtain
\[
\Prob_{\mathcal V_i^2}\left[
Y^S>\frac{100t_{i-1}'}{r}
\;\middle|\;
\mathcal V_i^1,\ X_i^2(u)=\xi_u
\right]
\le \left(\frac{e}{K}\right)^{100t_{i-1}'/r}
\le d^{-15}
\]
again because
\(t_{i-1}'\ge T=\Delta^{\eps_1}/3\), \(r=(\log\Delta)^{20}\), and
\(d\le \Delta^4\).

Now assume that \(\widetilde H_i(u,c)\) and \(J_{i,j}(u,c)\) do not occur.
Then \(|\Delta_2^S|\le |A_{i,j}(u,c)|\le 100t_{i-1}'/r\). Also,
\[
\Prob_{\mathcal V_i^2}\left[
(\widehat H_i(u,c))^\star
\mid
\mathcal V_i^1,\neg\widetilde H_i(u,c)
\right]
<d^{-15}.
\]
On the event
\[
\neg(\widehat H_i(u,c))^\star
\cap
\left\{Y^S\le \frac{100t_{i-1}'}{r}\right\},
\]
the three effects give
\begin{align*}
|\widehat N_{i,c}^S(u)|
&\le
|\widehat N_{i,c}^\star(u)|+|\Delta_1^S|+|\Delta_2^S|+|\Delta_3^S| \\
&\le
t_{i-1}'(1+\delta)\alpha_i\beta_i'+\delta T
+ \frac{100t_{i-1}'}{r}
+ \frac{100t_{i-1}'}{r}
+ \frac{100t_{i-1}'}{r} \\
&=
t_{i-1}'\left((1+\delta)\alpha_i\beta_i'+\frac{300}{r}\right)+\delta T.
\end{align*}
Therefore \((\widehat H_i^+(u,c))^S\) can occur only if either
\((\widehat H_i(u,c))^\star\) occurs or \(Y^S>100t_{i-1}'/r\). Since
\(\xi_u\) was arbitrary, averaging over \(X_i^2(u)\) when \(u\in S\) and taking
a union bound
gives
\[
\Prob_{\mathcal V_i^2}\left[
(\widehat H_i^+(u,c))^S
\mid
\mathcal V_i^1,\ \neg\widetilde H_i(u,c),\ \neg J_{i,j}(u,c)
\right]
\le 2d^{-15}<d^{-10}.
\]
Finally,
\[
\Prob_{\mathcal V_i^1}\left[
\widetilde H_i(u,c)\cup J_{i,j}(u,c)
\right]
\le 3d^{-35}<d^{-30}.
\]
Thus \(\widehat H_i^+(u,c)\), and hence \(H_i^+(u,c)\), is resilient with
respect to \(\Phi\).
\end{proof}

\section{Conclusions}

We have shown improved distributed algorithms for coloring triangle-free graphs, by improving the distributed Lovász Local Lemma steps that were previously the bottleneck in the algorithm of Pettie and Su \cite{PettieSu15}. This was achieved by adapting the algorithm and analysis of Pettie and Su to fit the LLL calls into the resilient partition framework of \cite{Davies23}. The LLL instances we solved here are arguably the most complex and difficult such instances yet to be solved in $\log^{O(1)}\log n$ rounds of \textsc{LOCAL}, and mark another step towards efficient distributed algorithms for the general LLL. For future work there are several interesting open questions:

\subsection{General LLL.}
Despite recent progress, an exponential gap still remains in the complexity of the general distributed LLL. Can a $\log^{O(1)}\log n$ round algorithm that solves all LLL instances (perhaps with polynomially weakened criteria, as in \cite{FischerGhaffari17,Davies23}, or can the $\Omega(\log\log n)$-round lower bound of \cite{BrandtEtAl16} be improved?

\subsection{Special cases of the LLL.}
If no efficient general LLL algorithm can be found, can we further identify applications with efficiently-solvable LLL instances, such as in other vertex coloring or degree splitting problems? In particular, this work indicates that LLL instances in which the bad events are failure to satisfy some concentration bounds, and in which we have some slack in the amount of allowable error, are prime candidates for solving quickly using \cite{Davies23}.

\subsection{Triangle-free coloring.}
Molloy \cite{Molloy19} showed that the constant in the number of colors needed for coloring triangle-free coloring can be improved from $4$ to $1$, but no (nontrivial) distributed algorithm yet exists to match this bound. Can such an algorithm be devised, and if so (since it will almost certinly involve LLL calls as in \cite{Molloy19}) can the LLL instances therein be efficiently solved, similarly to this work?

\bibliographystyle{plainurl}
\bibliography{refs}

@article{PettieSu15,
  author = {Seth Pettie and Hsin-Hao Su},
  title = {Distributed coloring algorithms for triangle-free graphs},
  journal = {Information and Computation},
  volume = {243},
  pages = {263--280},
  year = {2015},
  doi = {10.1016/j.ic.2014.12.018}
}

@misc{JamallColoring11,
  author = {Mohammad Shoaib Jamall},
  title = {A Coloring Algorithm for Triangle-Free Graphs},
  year = {2011},
  note = {Preprint}
}

@article{AjtaiKomlosSzemeredi80,
  author = {Mikl{\'o}s Ajtai and J{\'a}nos Koml{\'o}s and Endre Szemer{\'e}di},
  title = {A Note on Ramsey Numbers},
  journal = {Journal of Combinatorial Theory, Series A},
  volume = {29},
  number = {3},
  pages = {354--360},
  year = {1980},
  doi = {10.1016/0097-3165(80)90030-8}
}

@article{CHLPU19,
  title={Distributed edge coloring and a special case of the constructive Lov{\'a}sz local lemma},
  author={Chang, Yi-Jun and He, Qizheng and Li, Wenzheng and Pettie, Seth and Uitto, Jara},
  journal={ACM Transactions on Algorithms (TALG)},
  volume={16},
  number={1},
  pages={1--51},
  year={2019},
  publisher={ACM New York, NY, USA}
}

@inproceedings{Davies23,
  author = {Peter Davies},
  title = {Improved Distributed Algorithms for the Lov{\'a}sz Local Lemma and Edge Coloring},
  booktitle = {Proceedings of the 2023 Annual ACM-SIAM Symposium on Discrete Algorithms (SODA)},
  pages = {4273--4295},
  year = {2023},
  doi = {10.1137/1.9781611977554.ch163}
}

@inproceedings{Davies-Peck25,
  title={On the Locality of the Lov{\'a}sz Local Lemma},
  author={Davies-Peck, Peter},
  booktitle={Proceedings of the 57th Annual ACM Symposium on Theory of Computing},
  pages={1271--1282},
  year={2025}
}

@incollection{Kostochka78,
  author = {A. V. Kostochka},
  title = {Degree, girth and chromatic number},
  booktitle = {Combinatorics (Proc. Fifth Hungarian Colloq., Keszthely, 1976), Vol. II},
  series = {Colloq. Math. Soc. J{\'a}nos Bolyai},
  volume = {18},
  pages = {679--696},
  publisher = {North-Holland},
  year = {1978}
}

@article{Bollobas78,
  author = {B{\'e}la Bollob{\'a}s},
  title = {Chromatic number, girth and maximal degree},
  journal = {Discrete Mathematics},
  volume = {24},
  number = {3},
  pages = {311--314},
  year = {1978},
  doi = {10.1016/0012-365X(78)90102-4}
}

@article{GrablePanconesi00,
  author = {David A. Grable and Alessandro Panconesi},
  title = {Fast distributed algorithms for {Brooks-Vizing} colorings},
  journal = {Journal of Algorithms},
  volume = {37},
  number = {1},
  pages = {85--120},
  year = {2000}
}

@article{ChungPettieSu17,
  title={Distributed algorithms for the Lov{\'a}sz local lemma and graph coloring},
  author={Chung, Kai-Min and Pettie, Seth and Su, Hsin-Hao},
  journal={Distrib. Comput},
  volume={30},
  pages={261--280},
  year={2017}
}

@inproceedings{FischerGhaffari17,
  author = {Manuela Fischer and Mohsen Ghaffari},
  title = {Sublogarithmic Distributed Algorithms for Lov{\'a}sz Local Lemma, and the Complexity Hierarchy},
  booktitle = {31st International Symposium on Distributed Computing (DISC 2017)},
  series = {LIPIcs},
  volume = {91},
  pages = {18:1--18:16},
  year = {2017},
  doi = {10.4230/LIPIcs.DISC.2017.18}
}

@inproceedings{GhaffariHarrisKuhn18,
  author = {Mohsen Ghaffari and David G. Harris and Fabian Kuhn},
  title = {On Derandomizing Local Distributed Algorithms},
  booktitle = {2018 IEEE 59th Annual Symposium on Foundations of Computer Science (FOCS)},
  pages = {662--673},
  year = {2018},
  doi = {10.1109/FOCS.2018.00069}
}

@inproceedings{RozhonGhaffari20,
  author = {V{\'a}clav Rozho{\v{n}} and Mohsen Ghaffari},
  title = {Polylogarithmic-Time Deterministic Network Decomposition and Distributed Derandomization},
  booktitle = {Proceedings of the 52nd Annual ACM SIGACT Symposium on Theory of Computing (STOC)},
  pages = {350--363},
  year = {2020},
  doi = {10.1145/3357713.3384298}
}

@inproceedings{GhaffariGrunauRozhon21,
  author = {Mohsen Ghaffari and Christoph Grunau and V{\'a}clav Rozho{\v{n}}},
  title = {Improved Deterministic Network Decomposition},
  booktitle = {Proceedings of the 2021 ACM-SIAM Symposium on Discrete Algorithms (SODA)},
  pages = {2211--2231},
  year = {2021},
  doi = {10.1137/1.9781611976465.131}
}

@article{ChangPettie19,
  author = {Yi-Jun Chang and Seth Pettie},
  title = {A Time Hierarchy Theorem for the {LOCAL} Model},
  journal = {SIAM Journal on Computing},
  volume = {48},
  number = {1},
  pages = {33--69},
  year = {2019},
  doi = {10.1137/17M1157957}
}

@article{Molloy19,
  author = {Michael Molloy},
  title = {The list chromatic number of graphs with small clique number},
  journal = {Journal of Combinatorial Theory, Series B},
  volume = {134},
  pages = {264--284},
  year = {2019},
  doi = {10.1016/j.jctb.2018.06.007}
}

@article{Bernshteyn19,
  author = {Anton Bernshteyn},
  title = {The Johansson-Molloy theorem for DP-coloring},
  journal = {Random Structures \& Algorithms},
  volume = {54},
  number = {4},
  pages = {653--664},
  year = {2019},
  doi = {10.1002/rsa.20811}
}

@article{Linial92,
  author = {Nathan Linial},
  title = {Locality in Distributed Graph Algorithms},
  journal = {SIAM Journal on Computing},
  volume = {21},
  number = {1},
  pages = {193--201},
  year = {1992},
  doi = {10.1137/0221015}
}

@inproceedings{BrandtEtAl16,
  author = {Sebastian Brandt and Orr Fischer and Juho Hirvonen and Barbara Keller and Tuomo Lempi{\"a}inen and Joel Rybicki and Jukka Suomela and Jara Uitto},
  title = {A Lower Bound for the Distributed Lov{\'a}sz Local Lemma},
  booktitle = {Proceedings of the 48th Annual ACM SIGACT Symposium on Theory of Computing},
  pages = {479--488},
  year = {2016},
  doi = {10.1145/2897518.2897570}
}

@article{Johansson99,
  author = {{\"O}jvind Johansson},
  title = {Simple distributed {$\Delta +1$-coloring} of graphs},
  journal = {Information Processing Letters},
  volume = {70},
  number = {5},
  pages = {229--232},
  year = {1999}
}

@inproceedings{BarenboimEPS12,
  author = {Leonid Barenboim and Michael Elkin and Seth Pettie and Johannes Schneider},
  title = {The Locality of Distributed Symmetry Breaking},
  booktitle = {2012 IEEE 53rd Annual Symposium on Foundations of Computer Science (FOCS)},
  pages = {321--330},
  year = {2012},
  doi = {10.1109/FOCS.2012.60}
}

@article{BarenboimElkinKuhn14,
  author = {Leonid Barenboim and Michael Elkin and Fabian Kuhn},
  title = {Distributed {$(\Delta+1)$}-Coloring in Linear (in {$\Delta$}) Time},
  journal = {SIAM Journal on Computing},
  volume = {43},
  number = {1},
  pages = {72--95},
  year = {2014},
  doi = {10.1137/12088848X}
}

@article{HarrisSchneiderSu18,
  author = {David G. Harris and Johannes Schneider and Hsin-Hao Su},
  title = {Distributed {$(\Delta+1)$}-Coloring in Sublogarithmic Rounds},
  journal = {Journal of the ACM},
  volume = {65},
  number = {4},
  pages = {19:1--19:21},
  year = {2018},
  doi = {10.1145/3178120}
}

@article{ChangLiPettie20,
  author = {Yi-Jun Chang and Wenzheng Li and Seth Pettie},
  title = {Distributed {$(\Delta+1)$}-Coloring via Ultrafast Graph Shattering},
  journal = {SIAM Journal on Computing},
  volume = {49},
  number = {3},
  pages = {497--539},
  year = {2020},
  doi = {10.1137/19M1249527}
}

@techreport{JohanssonTriangleFree96,
  author = {{\"O}jvind Johansson},
  title = {Asymptotic Choice Number for Triangle-Free Graphs},
  institution = {DIMACS},
  number = {91-5},
  year = {1996}
}

@article{Kim95,
  author = {Jeong Han Kim},
  title = {On Brooks' Theorem for Sparse Graphs},
  journal = {Combinatorics, Probability and Computing},
  volume = {4},
  number = {2},
  pages = {97--132},
  year = {1995},
  doi = {10.1017/S0963548300001638}
}

@article{MoserTardos10,
  author = {Robin A. Moser and G{\'a}bor Tardos},
  title = {A constructive proof of the general Lov{\'a}sz local lemma},
  journal = {Journal of the ACM},
  volume = {57},
  number = {2},
  pages = {11:1--11:15},
  year = {2010},
  doi = {10.1145/1667053.1667060}
}

@inproceedings{HKNT22,
  author = {Magn{\'u}s M. Halld{\'o}rsson and Fabian Kuhn and Alexandre Nolin and Tigran Tonoyan},
  title = {Near-optimal distributed degree+1 coloring},
  booktitle = {Proceedings of the 54th Annual ACM SIGACT Symposium on Theory of Computing (STOC)},
  pages = {450--463},
  year = {2022},
  doi = {10.1145/3519935.3520023}
}

@inproceedings{HMN22,
  title={Fast Distributed Vertex Splitting with Applications},
  author={Halld{\'o}rsson, Magn{\'u}s M and Maus, Yannic and Nolin, Alexandre},
  booktitle={36th International Symposium on Distributed Computing (DISC 2022)},
  pages={26--1},
  year={2022},
  organization={Schloss Dagstuhl--Leibniz-Zentrum f{\"u}r Informatik}
}

@inproceedings{HM24,
  title={Distributed delta-coloring under bandwidth limitations},
  author={Halld{\'o}rsson, Magn{\'u}s M and Maus, Yannic},
  booktitle={38th International Symposium on Distributed Computing (DISC 2024)},
  pages={31--1},
  year={2024},
  organization={Schloss Dagstuhl--Leibniz-Zentrum f{\"u}r Informatik}
}

@INPROCEEDINGS{GG24,
  author={Ghaffari, Mohsen and Grunau, Christoph},
  booktitle={2024 IEEE 65th Annual Symposium on Foundations of Computer Science (FOCS)}, 
  title={Near-Optimal Deterministic Network Decomposition and Ruling Set, and Improved MIS}, 
  year={2024},
  volume={},
  number={},
  pages={2148-2179},
  doi={10.1109/FOCS61266.2024.00007}}
\appendix

\section{Auxiliary One-Round Concentration Bounds}
\label{sec:appendix-degree-tail}

Throughout this appendix, after fixing $\epsilon>0$, we assume
\(\Delta\ge \Delta_{\eps}\) is large enough for all asymptotic inequalities
below.

This appendix records the one-round concentration bounds used in Section~\ref{sec:triangle-free-instance-proof} and in Section~\ref{sec:g5-resilience}.
The first lemma concerns triangle-free graphs from
Section~\ref{sec:triangle-free-resilience}, where the relevant one-round events
are \(Q_i(u)\), \(H_i(u)\), and \(M_i(u)\). The second lemma concerns the
girth-\(5\) graphs from Section~\ref{sec:g5-resilience}.

\begin{lemma}
\label{lem:appendix-degree-tail}
Assume \(G\) is triangle-free, and fix a round \(i\) of
Algorithm~\ref{alg:triangle-free}. Suppose \(F_{i-1}\) holds and
\(d_{i-1}'>d_{\mathrm{stop}}\).
Then, for every active vertex \(u\),
\[
\Prob[Q_i(u)],\ \Prob[H_i(u)],\ \Prob[M_i(u)]\le \Delta^{-300},
\]
\end{lemma}

\begin{proof}
Fix an active vertex \(u\). In our notation, the palette lower-tail event
\(Q_i(u)\) is exactly the event treated in \cite[Lemma~5]{PettieSu15}, and the
averaged \(c\)-degree event \(H_i(u)\) is exactly the event treated in
\cite[Lemma~6]{PettieSu15}. Therefore
\[
\Prob[Q_i(u)],\ \Prob[H_i(u)]\le \exp\bigl(-\Omega(\delta^2p_i')\bigr).
\]

For \(M_i(u)\), let \(Z_u:=|N_i(u)|=\sum_{v\in N_{i-1}(u)} I_v\), where
\(I_v:=\mathbf 1_{\{v\in U_i\}}\), and expose
\[
\mathcal E_u:=\bigcap_{v\in N_{i-1}(u)} \neg Q_i(v),
\qquad
\mathcal G:=\sigma\bigl(X_i(x):x\notin N_{i-1}(u),\ K_i(v):v\in N_{i-1}(u)\bigr).
\]
Because \(G_{i-1}\) is triangle-free, vertices of \(N_{i-1}(u)\) are pairwise non-adjacent. Hence, for each \(v\in N_{i-1}(u)\), once \(\mathcal G\) is fixed, the indicator \(I_v\) depends only on \(X_i(v)\); therefore \(\{I_v:v\in N_{i-1}(u)\}\) are independent Bernoulli variables conditioned on \(\mathcal G\).

By the bound for \(Q_i(\cdot)\) established above and a union bound over
\(|N_{i-1}(u)|\le d_{i-1}'\le 2\Delta\), we have
\[
\Prob[\neg \mathcal E_u]\le \exp\bigl(-\Omega(\delta^2p_i')\bigr).
\]
Now fix a realization of \(\mathcal G\) for which \(\mathcal E_u\) holds. Then
for every \(v\in N_{i-1}(u)\),
\[
|\widetilde P_i(v)|\ge |P_{i-1}(v)|\beta_i\Bigl(1-\frac{\delta}{8}\Bigr).
\]
Also, Fact~\ref{fact:palette-lower} from
Section~\ref{sec:triangle-free-lemma-proofs} gives
\(|P_{i-1}(v)|\ge \left(\frac12-o(1)\right)p_{i-1}\), and
\(p_i'=(1-\delta/7)^i\beta_i p_{i-1}\) with \(i=O(\log\Delta)\). Therefore,
\[
|P_{i-1}(v)|\beta_i\Bigl(1-\frac{\delta}{8}\Bigr)
\ge
\left(1-\frac{(1+\delta)^{i-1}}{2}\right)p_i'.
\]
Since \(\widetilde P_i(v)\) is \(\mathcal G\)-measurable and
\(\mathrm{Sel}_i(v)\) is independent of \(\mathcal G\), we obtain
\[
\Prob[I_v=1\mid \mathcal G]
=
(1-\pi_i)^{|\widetilde P_i(v)|}
\le \alpha_i.
\]
Hence \(\mu_u:=\E[Z_u\mid \mathcal G,\mathcal E_u]\le \alpha_i d_{i-1}'\). Applying
the Chernoff bound conditional on \(\mathcal G\), with \(\eta=\delta/40\),
we obtain
\[
\Prob\left[M_i(u)\mid \mathcal G,\mathcal E_u\right]
\le \Prob\left[Z_u>(1+\eta)\mu_u\mid \mathcal G,\mathcal E_u\right]
\le \exp\bigl(-\Omega(\delta^2 d_{i-1}')\bigr),
\]
since \((1+\eta)\mu_u\le \alpha_i(1+\delta/5)d_{i-1}'\). Averaging over
\(\mathcal G\) and adding the failure probability of
\(\mathcal E_u\) yields
\[
\Prob[M_i(u)]
\le
\exp\bigl(-\Omega(\delta^2p_i')\bigr)+\exp\bigl(-\Omega(\delta^2 d_{i-1}')\bigr).
\]

As \(\delta=1/\log^4\Delta\),
Theorem~\ref{thm:executed-round-bounds} gives \(p_i'=\Omega(T)\) with
\(T=\Delta^{\eps_1/3}\), and \(d_{i-1}'>d_{\mathrm{stop}}=\Delta/(\log\Delta)^6\).
Hence \(\delta^2p_i'=\Omega(\Delta^{\eps_1/3}/\log^8\Delta)\) and
\(\delta^2d_{i-1}'\ge \Delta/(\log\Delta)^{14}\), so both exponents dominate
\(302\log\Delta\). 

Therefore \(\Prob[Q_i(u)]\), \(\Prob[H_i(u)]\), and
\(\Prob[M_i(u)]\le \Delta^{-300}\).
\end{proof}

\begin{lemma}
\label{lem:appendix-g5-degree-tail}
Assume \(G\) has girth at least \(5\), and fix a round \(i\) of
Algorithm~\ref{alg:g5-framework}. Suppose \(F_{i-1}\) holds and
\(d_{i-1}'>d_{\mathrm{stop}}\).
For every active vertex \(u\) and color \(c\in P_{i-1}(u)\), let
\[
\widehat N_{i,c}(u):=\{x\in N_{i-1,c}(u)\cap U_i:c\in \widehat P_i(x)\}.
\]
Then
\[
\Prob\left[
|\widehat N_{i,c}(u)|>t_{i-1}'(1+\delta)\alpha_i\beta_i'+\delta T
\right]
\le \Delta^{-200},
\]
\end{lemma}

\begin{proof}
In the notation of this paper, \(|\widehat N_{i,c}(u)|\) is exactly the quantity
controlled in \cite[Lemma~10]{PettieSu15}, and \cite[Corollary~5]{PettieSu15}
therefore gives
\[
\Prob\left[
|\widehat N_{i,c}(u)|>t_{i-1}'(1+\delta)\alpha_i\beta_i'+\delta T
\right]
\le \exp\bigl(-\Omega(\delta^2T)\bigr).
\]
In our settings, \(\delta=1/\log^4\Delta\) and
\(T=\Delta^{\eps_1}/3\), so \(\delta^2T=\Theta(\Delta^{\eps_1}/\log^8\Delta)\gg
\log\Delta\). Hence \(\exp(-\Omega(\delta^2T))\le \Delta^{-200}\), which proves
the lemma.
\end{proof}
\end{document}